\documentclass[12pt,reqno]{amsart}

\newcommand\version{October 26, 2022}

\usepackage{amsmath,amsfonts,amsthm,amssymb,amsxtra}
\usepackage{bbm} 
\usepackage{mathabx} 

\newtheorem{theorem}{Theorem}
\newtheorem{proposition}[theorem]{Proposition}
\newtheorem{lemma}[theorem]{Lemma}
\newtheorem{corollary}[theorem]{Corollary}

\theoremstyle{definition}

\theoremstyle{remark}

\newtheorem{remark}[theorem]{Remark}

\newcommand{\1}{\mathbbm{1}}

\newcommand{\C}{\mathbb{C}}

\newcommand{\D}{\mathbb{D}}
\renewcommand{\epsilon}{\varepsilon}

\newcommand{\N}{\mathbb{N}}

\renewcommand{\phi}{\varphi}
\newcommand{\R}{\mathbb{R}}

\newcommand{\Sph}{\mathbb{S}}

\newcommand{\Z}{\mathbb{Z}}

\DeclareMathOperator{\im}{Im}

\DeclareMathOperator{\Tr}{Tr}
\DeclareMathOperator{\tr}{Tr}

\begin{document}

\title[Sharp inequalities for coherent states --- \version]{Sharp inequalities for coherent states\\ and their optimizers}

\author{Rupert L. Frank}
\address[Rupert L. Frank]{Mathe\-matisches Institut, Ludwig-Maximilans Universit\"at M\"unchen, The\-resienstr.~39, 80333 M\"unchen, Germany, and Munich Center for Quantum Science and Technology, Schel\-ling\-str.~4, 80799 M\"unchen, Germany, and Mathematics 253-37, Caltech, Pasa\-de\-na, CA 91125, USA}
\email{r.frank@lmu.de}

\renewcommand{\thefootnote}{${}$} \footnotetext{\copyright\, 2022 by the author. This paper may be reproduced, in its entirety, for non-commercial purposes.\\
	Partial support through US National Science Foundation grant DMS-1954995, as well as through the Deutsche Forschungsgemeinschaft (German Research Foundation) through Germany’s Excellence Strategy EXC-2111-390814868 is acknowledged.}

\dedicatory{To David Jerison, in admiration, on the occasion of his 70th birthday}

\begin{abstract}
	We are interested in sharp functional inequalities for the coherent state transform related to the Wehrl conjecture and its generalizations. This conjecture was settled by Lieb in the case of the Heisenberg group and then by Lieb and Solovej for SU(2) and by Kulikov for SU(1,1) and the affine group. In this paper, we give alternative proofs and characterize, for the first time, the optimizers in the general case. We also extend the recent Faber--Krahn-type inequality for Heisenberg coherent states, due to Nicola and Tilli, to the SU(2) and SU(1,1) cases. Finally, we prove a family of reverse H\"older inequalities for polynomials, conjectured by Bodmann. 
\end{abstract}

\maketitle


\section{Introduction and main results}\label{sec:intro}

Coherent states appear in various areas of pure and applied mathematics, including mathematical physics, signal and image processing, semiclassical and microlocal analysis. Some background can be found, for instance, in \cite{Pe,Si,Sc}. Here we are interested in sharp functional inequalities for coherent state transforms.

To motivate the questions we are interested in let us recall Wehrl's conjecture \cite{We} and its resolution by Lieb \cite{Li}. Following Schr\"odinger, Bargmann, Segal, Glauber and others we consider a certain family of normalized Gaussian functions $\psi_{p,q}\in L^2(\R)$, parametrized by $p,q\in\R$. Explicitly,
$$
\psi_{p,q}(x) := (\pi\hbar)^{-\frac14}\, e^{-\frac1{2\hbar}(x-q)^2+ \frac 1\hbar ipx}
\qquad\text{for all}\ x\in\R \,,
$$
where $\hbar>0$ is a fixed constant. For a nonnegative operator $\rho$ in $L^2(\R)$ with $\tr\rho=1$ one considers the function
$$
(p,q) \mapsto \langle \psi_{p,q}, \rho\psi_{p,q} \rangle \,,
$$
known as Husimi function, covariant symbol or lower symbol. Thus, to a quantum state $\rho$ in $L^2(\R)$ one associates a function defined on the classical phase space $\R^2$. Wehrl \cite{We} was interested in the entropy-like quantity
$$
- \iint_{\R\times\R} \langle \psi_{p,q}, \rho\psi_{p,q} \rangle \ln \langle \psi_{p,q}, \rho\psi_{p,q} \rangle \,dp\,dq \,,
$$
showed that it is positive and conjectured that its minimum value occurs when $\rho=|\psi_{p_0,q_0}\rangle\langle\psi_{p_0,q_0}|$ for some $p_0,q_0\in\R$. That this is indeed the case was shown in an celebrated paper by Lieb \cite{Li}. Lieb's proof was based on the sharp form of the Young and the Hausdorff--Young inequalities and showed more generally that, for power functions $\Psi(s)=s^r$ with $r\geq 1$, the quantity
\begin{equation}
	\label{eq:phiquantity}
	\iint_{\R\times\R} \Phi(\langle \psi_{p,q}, \rho\psi_{p,q} \rangle) \,dp\,dq
\end{equation}
is maximal for $\rho$ as above. The result for $\Phi(s)=s\ln s$ then follows by differentiating at $r=1$, noting that the value for $\Phi(s)=s$ is a constant independent of $\rho$.

In \cite{Ca} Carlen gave an alternative proof of Lieb's result, both for $\Phi(s)=s^r$, $r\geq 1$, and $\Phi(s)=s\ln s$, and characterized the cases of equality. He also extended the result to $\Phi(s)=-s^r$ with $0<r<1$, again including a characterization of cases of equality. Carlen's proof is based on the logarithmic Sobolev inequality and an identity for analytic functions. For yet another proof in the logarithmic case see \cite{Lu}. For an interesting recent generalization of Lieb's result see \cite{dP}.

In \cite{LiSo1}, Lieb and Solovej extended the earlier results and showed what they called the generalized Wehrl conjecture. Namely, for \emph{any} convex function $\Phi$ on $[0,1]$ the quantity \eqref{eq:phiquantity} is maximal if $\rho=|\psi_{p_0,q_0}\rangle\langle\psi_{p_0,q_0}|$ for some $p_0,q_0\in\R$. The Lieb--Solovej proof proceeds by a limiting argument, based on sharp inequalities for SU(2) coherent states discussed below. Because of the limiting process, it does not provide a characterization of the cases of equality. Carlen's analysis is based on differentiating the power function $\Phi(s)=s^r$ with respect to the exponent $r$ and using the logarithmic Sobolev inequality for the resulting quantity. We do not know how to adopt this method to deal with general convex functions $\Phi$.

Our first main result in this paper gives an alternative proof of the theorem of Lieb and Solovej and includes a new characterization of the cases of equality.

\begin{theorem}\label{heisen}
	Let $\Phi:[0,1]\to\R$ be convex. Then
	$$
	\sup\left\{ \iint_{\R\times\R} \Phi( |\langle \psi_{p,q}, \psi \rangle|^2 ) \,dp\,dq :\ \psi\in L^2(\R) \,,\ \|\psi\|_{L^2(\R)} = 1 \right\} = 2\pi\hbar \int_0^1 \Phi(s) \,\frac{ds}{s}
	$$
	and the supremum is attained for $\psi=e^{i\theta} \psi_{p_0,q_0}$ with some $p_0,q_0\in\R$, $\theta\in\R/2\pi\Z$. If $\Phi$ is not linear and if the supremum is finite, then it is attained \emph{only} for such $\psi$.
\end{theorem}

Note that the value of the double integral with $\psi=e^{i\theta} \psi_{p_0,q_0}$ does not depend on $p_0,q_0\in\R$, $\theta\in\R/2\pi\Z$. It may or may not be finite, depending on $\Phi$. For finiteness it is necessary that $\lim_{s\to 0^+}\Phi(s)=0$.

Under a slightly stronger assumption on $\Phi$, we can extend the characterization of cases of equality to density matrices.

\begin{corollary}\label{heisencor}
	Let $\Phi:[0,1]\to\R$ be convex. Then
	$$
	\sup\left\{ \iint_{\R\times\R} \Phi(\langle \psi_{p,q}, \rho\psi_{p,q} \rangle) \,dp\,dq :\ \rho\geq 0 \ \text{on}\ L^2(\R) \,,\ \Tr\rho= 1 \right\} = 2\pi\hbar \int_0^1 \Phi(s) \,\frac{ds}{s}
	$$
	and the supremum is attained for $\rho=|\psi_{p_0,q_0}\rangle\langle\psi_{p_0,q_0}|$ with some $p_0,q_0\in\R$. If $\Phi$ is strictly convex and if the supremum is finite, then it is attained \emph{only} for such $\rho$.
\end{corollary}

\begin{remark}
	The statement and proof of Theorem \ref{heisen} and Corollary \ref{heisencor} extend, with minor changes, to the case of higher dimensions. We omit the details.
\end{remark}

Coherent states are often closely related to representations of an underlying Lie group. The coherent states discussed so far are related to the Heisenberg group. In his paper containing the proof of Wehrl's conjecture, Lieb conjectured that the analogue of Wehrl's conjecture also holds for Bloch coherent states, that is, for a family of coherent states related to SU(2). After some partial results in \cite{Sc0,Bo}, this conjecture was finally solved by Lieb and Solovej in \cite{LiSo1}; see also \cite{LiSo2} for a partially alternate proof. Again, they prove a generalized version of Lieb's conjecture involving general convex functions $\Phi$. However, they employ a limiting argument and therefore their paper does not characterize the cases of equality. Our second main result settles this open question by showing that, indeed, equality is only attained by rank one projections onto a coherent state.

Let us be more specific. As is well known (see, e.g., \cite[Chapter II]{Kn} and \cite[Section VIII.4]{Si2}), the nontrivial irreducible representations of SU(2) are labeled by $J\in\frac12\N$, where $2J+1$ is the dimension of the representation. Let $\mathcal H$ be a $(2J+1)$-dimensional representation space. Then there are operators $S_1,S_2,S_3$ on $\mathcal H$ satisfying $[S_1,S_2]=iS_3$ and cyclically, representing the generators of SU(2). For any $\omega\in\Sph^2$, the operator $\omega\cdot S=\omega_1 S_1 + \omega_2 S_2 + \omega_3 S_3$ has minimal eigenvalue $-J$ and this eigenvalue is nondegenerate. We choose $\psi_\omega\in\mathcal H$ as a corresponding normalized eigenvector. It is uniqe up to a phase, but since we are only interested in the state $|\psi_\omega\rangle\langle\psi_\omega|$, this choice of the phase is irrelevant for us. This defines the Bloch coherent states. (We follow here the convention in \cite{Pe}; other definitions are based on the \emph{maximal} eigenvalue $+J$, but this leads to the same family of coherent states, just interchanging $\omega$ and $-\omega$.)

\begin{theorem}\label{su2}
	Let $J\in\frac12\N$ and consider an irreducible $(2J+1)$-dimensional representation of $\mathrm{SU(2)}$ on $\mathcal H$. Let $\Phi:[0,1]\to\R$ be convex. Then
	$$
	\sup\left\{ \int_{\Sph^2} \Phi( |\langle \psi_\omega, \psi \rangle|^2) \,d\omega :\ \psi \in \mathcal H \,,\ \| \psi \|_{\mathcal H} = 1 \right\} = \frac{4\pi}{2J} \int_0^1 \Phi(s) s^{\frac{1}{2J}-1}\,ds
	$$
	and the supremum is attained for $\psi=e^{i\theta} \psi_{\omega_0}$ with some $\omega_0\in\Sph^2$, $\theta\in\R/2\pi\Z$. If $\Phi$ is not affine linear, then it is attained \emph{only} for such $\psi$.
\end{theorem}

Note that the value of the integral with $\psi=e^{i\theta} \psi_{\omega_0}$ does not depend on $\omega_0\in\Sph^2$, $\theta\in\R/2\pi\Z$. Since $\Phi$ is bounded, the supremum in the theorem is always finite, in contrast to Theorem \ref{heisen}.

\begin{corollary}\label{su2cor}
	Let $J\in\frac12\N$ and consider an irreducible $(2J+1)$-dimensional representation of $\mathrm{SU(2)}$ on $\mathcal H$. Let $\Phi:[0,1]\to\R$ be convex. Then
	$$
	\sup\left\{ \int_{\Sph^2} \Phi(\langle \psi_\omega, \rho\psi_\omega \rangle) \,d\omega :\ \rho\geq 0 \ \text{on}\ \mathcal H \,,\ \Tr\rho= 1 \right\} = \frac{4\pi}{2J} \int_0^1 \Phi(s) s^{\frac{1}{2J}-1}\,ds
	$$
	and the supremum is attained for $\rho=|\psi_{\omega_0}\rangle\langle\psi_{\omega_0}|$ with some $\omega_0\in\Sph^2$. If $\Phi$ is strictly convex, then it is attained \emph{only} for such $\rho$.
\end{corollary}

Our third main result concerns coherent states for certain representations of SU(1,1). After initial results in \cite{Ba,LiSo3,Baetal}, the analogue of Wehrl's conjecture was settled recently by Kulikov in \cite{Ku}, again for general convex functions $\Phi$. (In fact, slightly less than convexity is required in \cite{Ku}.) Kulikov \cite[Remark 4.3]{Ku} also characterizes optimizers in the case where $\Phi$ is strictly convex. We extend this to the case where $\Phi$ is not linear.

All nontrivial representations of SU(1,1) are infinite-dimensional. Its nontrivial irreducible unitary representations consist of discrete, principal and complementary series, as well as limits of the discrete series; see, e.g., \cite[Chapters II and XVI; also (2.20)]{Kn}. Here we are only interested in one of the two discrete series. The results for the other one can be deduced from the results below by complex conjugation at the appropriate places. 

Following the notation in \cite{Bar}, the discrete series representation under consideration is labeled by $K\in\frac12\N\setminus\{\frac12\}=\{ 1,\frac32,2,\ldots\}$. Let $\mathcal H$ be a corresponding representation space. The generators of the Lie algebra of SU(1,1) give rise to operators $K_0, K_1, K_2$ in $\mathcal H$ satisfying
$$
[K_1,K_2]=-iK_0 \,,\quad
[K_2,K_0]=iK_1 \,,\quad
[K_0,K_1]=iK_2 \,.
$$
Moreover, one has 
$$
K_0^2-K_1^2-K_2^2 = K(K-1) \,,
$$
where $K$ is the number labeling the representations. (There are also representations of SU(1,1) corresponding to $K=\frac12$, called limits of the discrete series, but their coherent state transforms are in some sense degenerate; see Subsection \ref{sec:limit}. We also briefly discuss the case of arbitrary real $K>\frac12$ after Corollary \ref{su11cor}.)

For any $(n_0,n_1,n_2)\in\R^3$ with $n_0^2-n_1^2-n_2^2=1$ and $n_0>0$ the operator $n_0 K_0 - n_1 K_1-n_2 K_2$ has minimal eigenvalue $K$ and this eigenvalue is simple. Therefore we can choose a corresponding normalized eigenvector (which is unique up to a phase). It is convenient to label this vector not by $(n_0,n_1,n_2)$ by rather by $z\in\D$, the open unit disk in $\C$, using the parametrization
$$
(n_0,n_1+ i n_2) = ( \tfrac{1+|z|^2}{1-|z|^2}, \tfrac{2z}{1-|z|^2}) \,.
$$
In this way we obtain a family of vectors $\psi_z$, $z\in\D$, giving rise to coherent states for a discrete series representation of SU(1,1).

In what follows, we denote by $dA(z)=dx\,dy$ the two-dimensional Lebesgue measure on $\C$.

\begin{theorem}\label{su11}
	Let $K\in\frac12\N\setminus\{\tfrac12\}$ and consider the irreducible discrete series representation of $\mathrm{SU(1,1)}$ on $\mathcal H$ corresponding to $K$. Let $\Phi:[0,1]\to\R$ be convex. Then
	$$
	\sup\left\{ \int_{\D} \Phi(|\langle \psi_z, \psi \rangle|^2) \,\frac{dA(z)}{(1-|z|^2)^2} :\ \psi\in\mathcal H \,,\ \|\psi\|_\mathcal H= 1 \right\} = \frac{\pi}{2K} \int_0^1 \Phi(s) s^{-\frac{1}{2K}-1}\,ds
	$$
	and the supremum is attained for $\psi=e^{i\theta}\psi_{z_0}$ with some $z_0\in\D$, $\theta\in\R/2\pi\Z$. If $\Phi$ is not linear and if the supremum is finite, then it is attained \emph{only} for such $\psi$.
\end{theorem}

Note that the value of the integral with $\psi=e^{i\theta}\psi_{z_0}$ does not depend on $z_0\in\D$, $\theta\in\R/2\pi\Z$. It may or may not be finite, depending on $\Phi$. For finiteness it is necessary that $\lim_{s\to 0^+}\Phi(s) = 0$.

Theorem \ref{su11} proves the uniqueness part of a conjecture of Lieb and Solovej \cite[Conjecture 5.2]{LiSo3}. As we mentioned before, the inequality part is due to Kulikov \cite{Ku}.

\begin{corollary}\label{su11cor}
	Let $K\in\frac12\N\setminus\{\tfrac12\}$ and consider the irreducible discrete series representation of $\mathrm{SU(1,1)}$ on $\mathcal H$ corresponding to $K$. Let $\Phi:[0,1]\to\R$ be convex. Then
	$$
	\sup\left\{ \int_{\D} \Phi(\langle \psi_z, \rho\psi_z \rangle) \,\frac{dA(z)}{(1-|z|^2)^2} :\ \rho\geq 0 \ \text{on}\ \mathcal H \,,\ \Tr\rho= 1 \right\} = \frac{\pi}{2K} \int_0^1 \Phi(s) s^{-\frac{1}{2K}-1}\,ds
	$$
	and the supremum is attained for $\rho=|\psi_{z_0}\rangle\langle\psi_{z_0}|$ with some $z_0\in\D$. If $\Phi$ is strictly convex and if the supremum is finite, then it is attained \emph{only} for such $\rho$.
\end{corollary}

For every real $K>\frac12$, there is an irreducible representations of the Lie \emph{algebra} generated by $K_0, K_1,K_2$ satisfying the above relations. For $K\not\in\frac12\N$ such a representation does not come from the Lie \emph{group} SU(1,1) -- recall that this group is not simply connected. It does come, however, from a representation of the covering group of SU(1,1) \cite{Pe0} and we could prove sharp inequalities for the corresponding coherent states. From an analytic point of view this would lead to the same problem as coherent states for the affine group, which we discuss next.

The affine group (in one space dimension), also known as the $(aX+b)$-group, has two nontrivial, irreducible unitary representations \cite{AsKl,DaKlPa}. Again, we focus on a single one since the results for the other one can be obtained by appropriate complex conjugations. What distinguishes the affine group from the above cases of the Heisenberg group, SU(2) and SU(1,1) is that different choices of an extremal weight vector lead to inequivalent coherent state transforms.

We fix a parameter $\beta>\frac12$, emphasizing that this parameter does not label the representation, but rather the choice of the extremal weight vector. We consider the following family of normalized functions $\psi_{a,b}\in L^2(\R_+)$, $\R_+=(0,\infty)$, parametrized by $a\in\R_+$, $b\in\R$,
$$
\psi_{a,b}(x) :=  2^\beta \Gamma(2\beta)^{-\frac12}\ a^\beta x^{\beta-\frac12} e^{-ax+ibx}
\qquad\text{for all}\ x\in\R_+ \,.
$$
(Here we follow the convention in \cite{DaKlPa}. What we call $\beta$ is called $\alpha-\frac12$ in \cite{LiSo3} and they do not normalize the $\psi_{a,b}$ in $L^2(\R_+)$, but choose a different, natural normalization.)

\begin{theorem}\label{axb}
	Let $\beta>\frac12$ and let $\Phi:[0,1]\to\R$ be convex. Then
	$$
	\sup\left\{ \iint_{\R_+\times\R} \! \Phi( |\langle \psi_{a,b}, \psi \rangle|^2 ) \,\frac{da\,db}{a^2} :\ \psi\in L^2(\R_+) \,,\ \|\psi\|_{L^2(\R_+)} = 1 \right\} \! = \frac{2\pi}{\beta} \int_0^1 \Phi(s) s^{-\frac{1}{2\beta}-1}\,ds
	$$
	and the supremum is attained for $\psi=e^{i\theta} \psi_{a_0,b_0}$ with some $a_0\in\R_+,b_0\in\R$, $\theta\in\R/2\pi\Z$. If $\Phi$ is not linear and if the supremum is finite, then it is attained \emph{only} for such~$\psi$.
\end{theorem}

Note that the value of the double integral with $\psi=e^{i\theta} \psi_{a_0,b_0}$ does not depend on $a_0\in\R_+,b_0\in\R$, $\theta\in\R/2\pi\Z$. It may or may not be finite, depending on $\Phi$. For finiteness it is necessary that $\lim_{s\to 0^+}\Phi(s)=0$.

Theorem \ref{axb} settles the equality part of a conjecture of Lieb and Solovej \cite[Conjecture 3.1]{LiSo3}. For strictly convex $\Phi$ it had been settled earlier in \cite[Remark 4.3]{Ku}. Clearly, the assumption of $\Phi$ not being linear is optimal, since otherwise the supremum is attained for any $\psi\in L^2(\R_+)$.

We note that Theorem \ref{axb} has a version for $\beta=\frac12$; see Remark \ref{axblimit}.

\begin{corollary}\label{axbcor}
	Let $\beta>\frac12$ and let $\Phi:[0,1]\to\R$ be convex. Then
	$$
	\sup\left\{ \iint_{\R_+\times\R} \! \Phi(\langle \psi_{a,b}, \rho\psi_{a,b} \rangle) \,\frac{da\,db}{a^2} :\ \rho\geq 0 \ \text{on}\ L^2(\R) \,,\ \Tr\rho= 1 \right\} \!= \frac{2\pi}{\beta} \int_0^1 \Phi(s) s^{-\frac{1}{2\beta}-1}\,ds
	$$
	and the supremum is attained for $\rho=|\psi_{a_0,b_0}\rangle\langle\psi_{a_0,b_0}|$ with some $a_0\in\R_+,b_0\in\R$. If $\Phi$ is strictly convex and if the supremum is finite, then it is attained \emph{only} for such $\rho$.
\end{corollary}

This concludes the description of our main results, but we would like to draw the reader's attention also to Sections \ref{sec:reverse} and \ref{sec:fk} where we prove, respectively, sharp reverse H\"older inequalities for analytic functions, thereby settling a conjecture of Bodmann \cite{Bo}, and optimal Faber--Krahn-type inequalities for coherent state transforms.

The method that we will be using is that from a recent beautiful paper by Kulikov \cite{Ku}. He developed this method to solve the Lieb--Solovej conjectures for SU(1,1) and the affine group. Here we show that it can be adapted to deal with the Heisenberg and the SU(2) case. We also push the characterization of optimizers a bit further than in \cite{Ku}, thus leading to the optimal results in Theorems \ref{heisen}, \ref{su2}, \ref{su11} and \ref{axb}.

Kulikov's paper in turn seems to be inspired by an equally beautiful recent paper by Nicola and Tilli \cite{NiTi}. They were the first, as far as we know, to use the isoperimetric inequality in connection with the coherent state transform to obtain optimal functional inequalities. (Talenti \cite{Ta} used a closely related method for comparison theorems for solutions of PDEs.) Kulikov proved his results by using instead the isoperimetric inequality in hyperbolic space and we will prove Theorem \ref{su2} by using that on the sphere. While it is tempting to try to use the same method for more general groups, an obstacle would have to be overcome; see Subsection \ref{sec:limitation}.

Nicola and Tilli proved Faber--Krahn-type inequalities for the Heisenberg coherent states. We will show that their main result (at least without the characterization of the cases of equality) follows from Theorem \ref{heisen} and we will use this idea to prove analogues of their results for coherent states based on SU(2), SU(1,1) and the affine group; see Section \ref{sec:fk}. For further developments started by \cite{NiTi} see, for instance, \cite{RaTi,NiTi2,Ka}.

After this paper was submitted for publication, we learned that Aleksei Kulikov, Fabio Nicola, Joaquim Ortega-Cerd\'a and Paolo Tilli have independently obtained similar results with similar techniques.

\medskip

Thanks are due to Eric Carlen, Elliott Lieb and Jan Philip Solovej for many discussions on the topics of this paper.

\medskip

It is my pleasure to dedicate this paper to David Jerison on the occasion of his 70th birthday. His papers have been an inspiration for me, those on sharp inequalities \cite{JeLe} and others. I am particularly indebted to him for his remarks in the fall of 2008, which indirectly were a great motivation for work that eventually led to~\cite{FrLi}.
 

\section{Inequalities for analytic functions}\label{sec:analytic}

The main ingredient behind the results in the previous section are sharp inequalities for analytic functions and the characterization of their optimizers, which we discuss in the present section.


\subsection{Definitions and main result}

There are three different types of inequalities, corresponding to the cases of the Heisenberg group, SU(2) and SU(1,1). We refer to these different scenarios as Cases 1, 2 and 3. In Cases 2 and 3, there is a parameter $J\in\frac12\N$ and $\alpha>1$, respectively, that is fixed in what follows. 

In Case 1, we consider functions from the Fock space $\mathcal F^2(\C)$, that is, entire functions $f$ satisfying
$$
\| f \|_{\mathcal F^2} := \left( \int_\C |f(z)|^2 e^{-\pi |z|^2}\,dA(z) \right)^{1/2}<\infty \,.
$$
We recall that we write $dA(z)=dx\,dy$ for the two-dimensional Lebesgue measure. In Case~2, we consider functions in $\mathcal P_{2J}$, that is, polynomials $f$ of degree $\leq 2J$ endowed with the norm
$$
\| f \|_{\mathcal P_{2J}} := \left( \frac{2J+1}{\pi} \int_\C |f(z)|^2\, (1+|z|^2)^{-2J-2}\,dA(z) \right)^{1/2}.
$$
This norm is finite for any $f\in\mathcal P_{2J}$. In Case 3, we consider functions from the weighted Bergman space $A^2_\alpha(\D)$, that is, analytic functions $f$ on the disk $\D$ satisfying
$$
\| f \|_{A^2_\alpha} := \left( \frac{\alpha-1}\pi \int_\D |f(z)|^2 (1-|z|^2)^{\alpha-2}\,dA(z) \right)^{1/2}<\infty \,.
$$
To treat all cases simultaneously, we set
$$
\mathcal H :=
\begin{cases}
	\mathcal F^2(\C) & \text{in Case 1} \,,\\
	\mathcal P_{2J} & \text{in Case 2}\,, \\
	\mathcal A^2_\alpha(\D) & \text{in Case 3} \,,
\end{cases}
$$
and denote the norm in $\mathcal H$ by $\vvvert \cdot \vvvert$.

Thus, the set on which the relevant functions are defined is
$$
\Omega:=
\begin{cases}
	\C & \text{in Cases 1 and 2} \,, \\
	\D & \text{in Case 3} \,,
\end{cases}
$$
and the relevant measure is
$$
dm(z) :=
\begin{cases}
	dA(z) & \text{in Case 1} \,, \\
	\pi^{-1} (1+|z|^2)^{-2}\,dA(z) & \text{in Case 2}\,,\\
	\pi^{-1} (1-|z|^2)^{-2}\,dA(z) & \text{in Case 3} \,.
\end{cases}
$$
To a function $f\in\mathcal H$, we associate a function $u_f$ on $\Omega$, defined by
\begin{align}\label{eq:uf}
	u_f(z) & := 
	\begin{cases}
		|f(z)| e^{-\frac\pi2|z|^2}  & \text{in Case 1} \,,\\
		|f(z)| (1+|z|^2)^{-J} & \text{in Case 2} \,,\\
		|f(z)| (1-|z|^2)^{\alpha/2} & \text{in Case 3} \,.
	\end{cases}
\end{align}

The problem that we are interested in is to maximize, given a convex function $\Phi$, the quantity
$$
\int_\Omega \Phi( u_f(z)^2) \,dm(z)
$$
over all $f\in\mathcal H$ with $\vvvert f \vvvert =1$. Our main result will characterize the set of $f$'s for which this supremum is attained. This set $\mathcal M\subset\{ f\in\mathcal H:\ \vvvert f \vvvert =1 \}$ is defined as follows. In Case 1, we consider the functions $F_w$, parametrized by $w\in\C$, given by
$$
F_w(z) := e^{-\frac\pi2|w|^2+ \pi \overline w z} 
\qquad\text{for all}\ z\in\C \,,
$$
and set
$$
\mathcal M:=\left\{ e^{i\theta}\, F_w :\ w\in\C \,,\ \theta\in\R/2\pi\Z \right\}.
$$
In Case 2, we consider the functions $F_w$, parametrized by $w\in\C\cup\{\infty\}$, given by
$$
F_w(z) := \frac{(1+\overline w z)^{2J}}{(1+|w|^2)^J}\,,\ w\neq 0 \,,\qquad
F_\infty(z) := z^{2J}
\qquad\text{for all}\ z\in\C \,,
$$
and set
$$
\mathcal M:=\left\{ e^{i\theta}\, F_w :\ w\in\C\cup\{\infty\} \,,\ \theta\in\R/2\pi\Z \right\}.
$$
In Case 3, we consider the function $F_w$, parametrized by $w\in\D$, given by
$$
F_w(z):= \frac{(1-|w|^2)^\frac\alpha2}{(1-\overline wz)^\alpha}
\qquad\text{for all}\ z\in\D \,,
$$
and set
$$
\mathcal M:=\left\{ e^{i\theta}\, F_w :\ w\in\D \,,\ \theta\in\R/2\pi\Z \right\}.
$$
In each case it can be verified that $F_w\in\mathcal H$ and that $\vvvert F_w \vvvert=1$ for all $w$ in the respective index set. Indeed, this can be seen by a direct computation for $w=0$. For a general $w$ we use the fact that the functions $u_{F_w}$ are equimeasurable, that is, for every $\kappa>0$, the measure $m(\{ u_{F_w}>\kappa\})$ is independent of $w$. As a consequence of the equimeasurability the norm in $\mathcal H$ is independent of $w$. The equimeasurability in turn is a consequence of the fact that $F_w$ is obtained from, say, $F_0$ by the action of the Heisenberg group, SU(2) or SU(1,1) in the respective cases and of the invariance of the measure $m$ under this action.

The following is the main result of this section.

\begin{theorem}\label{analytic}
	Let $\Phi:[0,1]\to\R$ be convex. Then
	\begin{align*}
		& \sup\left\{ \int_\Omega \Phi( u_f(z)^2) \,dm(z) :\ f\in\mathcal H \,,\ \vvvert f\vvvert =1 \right\} \\
		& =
		\begin{cases}
			\int_0^1 \Phi(s) s^{-1}\,ds & \qquad\text{in Case 1}\,,\\
			(2J)^{-1} \int_0^1 \Phi(s) s^{\frac{1}{2J}-1}\,ds & \qquad\text{in Case 2}\,,\\
			\alpha^{-1} \int_0^1 \Phi(s) s^{-\frac{1}{\alpha}-1}\,ds & \qquad\text{in Case 3}\,,
		\end{cases}
	\end{align*}
	and the supremum is attained in $\mathcal M$. If $\Phi$ is not affine linear and if the supremum is finite, then it is attained \emph{only} in $\mathcal M$.
\end{theorem}

We will prove this theorem in the next subsection, after establishing some lemmas.


\subsection{Proof of Theorem \ref{analytic}}

We begin the proof of Theorem \ref{analytic} by recalling a simple and well-known bound on the supremum of $u_f$. This bound shows that $u_f\leq 1$ for $\vvvert f\vvvert =1$, so $\Phi(u_f^2)$ appearing in the theorem is well defined for $\Phi$ defined on $[0,1]$. The characterization of the cases of equality in the inequality $u_f\leq 1$ will eventually lead to the corresponding characterization in Theorem \ref{analytic}.

\begin{lemma}\label{linfty}
	Let $f\in\mathcal H$. Then
	$$
	\| u_f \|_{L^\infty(\Omega)} \leq \vvvert f \vvvert
	$$
	with equality if and only if either $f=0$ or $\vvvert f\vvvert^{-1} f\in\mathcal M$.
\end{lemma}

\begin{proof}
	In Case 1, this is essentially \cite[Proposition 2.1]{NiTi}. Indeed, there the inequality in the lemma is proved and it is shown that $u_f$ tends to zero at infinity. The latter fact, together with continuity, implies that there is an $z\in\C$ such that $\|u_f\|_{L^\infty(\C)}=u_f(z)$, and then \cite[Proposition 2.1]{NiTi} implies $\vvvert f\vvvert^{-1} f\in\mathcal M$, provided $f\neq 0$.
	
	In Case 2, the inequality is mentioned in \cite[Paragraph after Remark 3.2]{Bo}. Since the function $u_f$ extends continuously to a function on the Riemann sphere $\C\cup\{\infty\}$, there is a $z\in\C\cup\{\infty\}$ such that $\|u_f\|_{L^\infty(\C)}=u_f(z)$, and then one obtains the equality condition from that in the Cauchy--Schwarz inequality.
	
	In Case 3, the inequality and the fact that $u_f$ tends to zero as $|z|\to 1$ is mentioned in \cite[(1.1) and the paragraph thereafter]{Ku}. As in the other cases, from the latter fact one can deduce the equality condition. Let us add some details concerning the facts mentioned in \cite{Ku}. To carry out the proof of the inequality, one can, for instance, use \cite[(5)]{Baetal} and argue as in \cite[Proposition 2.1]{NiTi}. To deduce the vanishing of $u_f$ on $\partial\D$, one can observe that this is true when $f$ is a polynomial and that those are dense in $A^2_\alpha(\D)$ by \cite[(5)]{Baetal}.
\end{proof}

We now come to the core of the proof of Theorem \ref{analytic}, which concerns a monotonicity of the measure of superlevel sets of $u_f$. In Case 3, the following lemma and its proof are a special case of \cite[Theorem 2.1]{Ku}. Our contribution is to extend the reasoning to Cases 1 and 2. (Indeed, similar arguments in the setting of Case 1 have already appeared in \cite[Theorem 3.1]{NiTi}; in particular, inequality \eqref{eq:muprime} is the same as \cite[(3.17) combined with (3.10)]{NiTi}. Note, however, that the authors of \cite{NiTi} use Lieb's inequality as an ingredient and do not reprove it using their method, in contrast to what we do here; see, in particular, \cite[Theorem 5.2]{NiTi}.) We also point out that arguments of this type are reminiscent of those of Talenti in \cite{Ta}.

\begin{lemma}\label{mono}
	Let $f\in\mathcal H$ and $\mu(\kappa):= m(\{u_f>\kappa\})$ for $\kappa>0$. Then the function
	\begin{align*}
		\kappa \mapsto 
		\begin{cases}
			-2\ln \kappa - \mu(\kappa) & \text{in Case 1} \,,\\
			\kappa^{-\frac1J}(1-\mu(\kappa)) & \text{in Case 2} \,,\\
			\kappa^\frac2\alpha(-1-\mu(\kappa)) & \text{in Case 3} \,,
		\end{cases}
	\end{align*}
	is nondecreasing on $(0,\|u_f\|_\infty)$. Moreover, if $f\in\mathcal M$, then this function is constant.
\end{lemma}

\begin{proof}
	Writing $M:=\|u_f\|_{L^\infty(\Omega)}$ for brevity, we observe that, by Sard's theorem (noting the $u_f$ is real analytic on $\Omega$ viewed as a subset of $\R^2$), for almost every $\kappa\in(0,M)$, $\{ u_f =\kappa\}$ is a smooth curve (or, possibly, union thereof). Denote one-dimensional Hausdorff measure on this curve by $|dz|$, we have, by the co-area formula,
	$$
	\int_\Omega g(x,y) |\nabla u_f|\,dA(z) = \int_0^{t_0} \int_{\{u_f=\kappa\}} g(x,y)\,|dz| \,d\tau \,.
	$$
	Let
	$$
	\omega(z) := \frac{dm(z)}{dA(z)} =
	\begin{cases}
		1 & \qquad\text{in Case 1} \,,\\
		\pi^{-1} (1+|z|^2)^{-2} & \qquad\text{in Case 2} \,,\\
		\pi^{-1} (1-|z|^2)^{-2} & \qquad\text{in Case 3} \,.
	\end{cases}
	$$
	Taking $g = |\nabla u_f|^{-1}\omega \1_{\{u_f>\kappa\}}\1_{\{|\nabla u_f|\neq 0\}}$ and noting that $m(\{ |\nabla u_f|= 0\})=0$ by real analyticity, we find, for any $\kappa\in(0,M)$,
	$$
	\mu(\kappa) = \int_\kappa^{M}  \int_{\{u_f=\kappa\}} |\nabla u_f|^{-1} \omega \,|dz|\,d\tau \,.
	$$
	Thus, $\mu$ is absolutely continuous on compact subintervals of $(0,M]$ and for almost every $\kappa\in(0,M)$,
	$$
	\mu'(\kappa) = -  \int_{\{u_f =\kappa \}} |\nabla u_f|^{-1} \omega\,|dz| \,.
	$$
	For a curve $\gamma$ in $\Omega$ let us set
	$$
	\ell(\gamma): = \int_\gamma \sqrt\omega\,|dz| \,.
	$$
	In particular, for the level set $\{u_f=\kappa\}$ we obtain, by the Schwarz inequality,
	$$
	\ell(\{u_f=\kappa\})^2 \leq  \int_{\{u_f=\kappa\}} |\nabla u_f|^{-1} \omega\,|dz| \ \int_{\{u_f=\kappa\}} |\nabla u_f| |dz| \,.
	$$
	As we argued before, the first term on the right side is $-\mu'(\kappa)$. Let us consider the second term. Since the outer unit normal vector field $\nu$ to $\{u_f>\kappa\}$ on the boundary $\{u_f=\kappa\}$ is given by $-\nabla u_f/|\nabla u_f|$, we have $|\nabla u_f| = - \kappa \nu\cdot \nabla (\ln u_f)$ and therefore, by Green's theorem,
	$$
	\int_{\{u_f=\kappa\}} |\nabla u_f| |dz| = - \kappa \int_{\{u_f>\kappa \}} \Delta \ln u_f \,dA(z) \,.
	$$
	To compute the Laplacian of $\ln u_f$ we recall that $f$ is a positive weight times the absolute value of an analytic function. On the set $\{ u_f>\kappa\}$ the analytic function does not have zeros, so the logarithm of its absolute value is harmonic there. Thus, the Laplacian of $\ln u_f$ coincides with the Laplacian of the logarithm of the weight. Explicitly,
	$$
	\Delta \ln u_f =
	\begin{cases}
		-\frac\pi2 \Delta |z|^2 = -2\pi & \qquad\text{in Case 1} \,,\\
		-j \Delta\ln (1+|z|^2) = -4j (1+|z|^2)^{-2} & \qquad\text{in Case 2} \,,\\
		\frac\alpha 2 \Delta\ln(1-|z|^2) = -2\alpha (1-|z|^2)^{-2} & \qquad\text{in Case 3} \,. 
	\end{cases}
	$$
	Note that the right side is equal to a constant multiple of $\omega$ and therefore
	$$
	\int_{\{u_f=\kappa\}} |\nabla u_f| |dz| = 
	\begin{cases}
		2\pi \kappa \mu(\kappa) & \qquad\text{in Case 1}\,\\
		4\pi j \kappa \mu(\kappa) & \qquad\text{in Case 2}\,,\\
		2\pi \alpha \kappa\mu(\kappa) & \qquad\text{in Case 3} \,.
	\end{cases}
	$$
	To summarize, we have shown that
	$$
	\ell(\{u_f=\kappa\})^2 \leq 
	\begin{cases}
		-2\pi\kappa\mu'(\kappa)\mu(\kappa) & \qquad\text{in Case 1}\,\\
		-4\pi j \kappa \mu'(\kappa) \mu(\kappa) & \qquad\text{in Case 2}\,,\\
		- 2\pi \alpha \kappa \mu'(\kappa) \mu(\kappa) & \qquad\text{in Case 3} \,.
	\end{cases}
	$$
	We now use the isoperimetric inequality to bound the left side from below; for references in the spherical and hyperbolic case see, for instance, \cite[(4.23)]{Os}, as well as \cite{Be}, \cite[Third part, Chapter IV]{Le}, \cite{Ra,Sm1,Sm2}. We have
	$$
	\ell(\partial A)^2 \geq 
	\begin{cases}
		4\pi m(A) & \qquad\text{in Case 1}\,,\\
		4\pi m(A) (1-m(A)) & \qquad\text{in Case 2}\,,\\
		4\pi m(A) (1+ m(A)) & \qquad\text{in Case 3} \,.
	\end{cases}	
	$$
	Using these inequalities with $A=\{u_f>\kappa\}$, dividing by $\mu(\kappa)$ (which is nonzero for $\kappa<M$) and combining the resulting inequality with the above upper bound on $\ell(\{u_f=\kappa\})^2$ we obtain
	\begin{equation}
		\label{eq:muprime}
		\begin{cases}
		2 \leq -\kappa \mu'(\kappa) & \qquad\text{in Case 1}\,,\\
		1-\mu(\kappa) \leq - j\kappa \mu'(\kappa) & \qquad\text{in Case 2}\,,\\
		1+ \mu(\kappa) \leq - \frac\alpha 2\kappa \mu'(\kappa) & \qquad\text{in Case 3} \,.
		\end{cases}
	\end{equation}
	These inequalities are equivalent to the monotonicity assertions in the lemma.
	
	It remains to verify that this function is constant if $f\in\mathcal M$. By the equimeasurability discussed before the statement of Theorem \ref{analytic} it suffices to prove this for $f=F_0\in\mathcal M$. For all $\kappa\leq1$ we have
	$$
	m(\{ u_{F_0}>\kappa\}) = 
	\begin{cases}
		\int_{\C} \1(e^{-\frac\pi2|z|^2}>\kappa)\,dA(z) = -2\ln\kappa & \qquad\text{in Case 1}\,,\\
		\pi^{-1} \int_\C \1((1+|z|^2)^{-J}>\kappa)\,\frac{dA(z)}{(1+|z|^2)^2} = 1- \kappa^{\frac{1}{J}} & \qquad\text{in Case 2}\,,\\
		\pi^{-1} \int_\D \1((1-|z|^2)^{\frac\alpha2}>\kappa)\,\frac{dA(z)}{(1-|z|^2)^2} = \kappa^{-\frac2\alpha}-1 & \qquad\text{in Case 3}\,.
	\end{cases}
	$$
	It follows that for $f=F_0$ the function in Lemma \ref{mono} is, indeed, constant.
\end{proof}

The last ingredient in the proof of Theorem \ref{analytic} is an inequality due to Chebyshev \cite{Ch}; see also \cite[Theorems 43 and 236]{HaLiPo}. For a proof of the following lemma, with a slightly weaker assumption than monotonicity of one of the functions see \cite[Lemma~4.1]{Ku}.

\begin{lemma}\label{kulikov}
	Let $t_0>0$ and let $w,h$ be nondecreasing functions on $[0,t_0]$. Then
	$$
	\int_0^{t_0} h(t)w(t)\,dt \geq t_0^{-1} \int_0^{t_0} h(t)\,dt\ \int_0^{t_0} w(t)\,dt \,. 
	$$
\end{lemma}

We are finally in position to prove the main result of this section.

\begin{proof}[Proof of Theorem \ref{analytic}]
	We begin with some preliminary remarks concerning convex functions $\Phi$ on $[0,1]$. We first argue that without loss of generality we may assume that $\Phi$ is continuous on $[0,1]$. By convexity, it is continuous on $(0,1)$, so we only need to discuss the endpoints. It is elementary that $\Phi(0^+):=\lim_{s\to 0^+}\Phi(s)$ and $\Phi(1^-) := \lim_{s\to 1^-}\Phi(s)$ exist and are finite. (Note that these limits are $\leq\Phi(0)$ and $\leq\Phi(1)$, respectively, so in particular they are not $+\infty$.) By analyticity, $m(\{u_f=0\}) = m(\{ f=0\})=0$, so on this set we may replace $\Phi(0)$ by $\Phi(0^+)$ without changing the value of the integral. Similarly, by Lemma \ref{linfty} and its proof, $\{ u_f =1\}$ consists at most of one point, so on this set we may replace $\Phi(1)$ by $\Phi(1^-)$ without changing the value of the integral. Thus, we may assume that $\Phi$ is continuous on $[0,1]$.
	
	Next, we argue that we may assume that $\Phi(0)=0$. In Case 2, $m$ is a finite measure, so this can be accomplished by replacing $\Phi$ by $\Phi-\Phi(0)$, which has a trivial effect on the supremum. In Cases 1 and 3, we take $f\in\mathcal M$ and see from the explict form that $u_f(z)\to 0$ as $|z|\to\infty$ in Case 1 and $|z|\to 1$ in Case 3. (In fact, this holds for any $f\in\mathcal H$, as discussed in the proof of Lemma \ref{linfty}, but this is not needed here.) It follows that, if $\Phi(0)\neq 0$, then the supremum is equal to $+\infty$ and this value is achieved by all $f\in\mathcal M$, so the assertion of the theorem is true in this case. Thus, in what follows we may assume that $\Phi(0)=0$.
	
	After these preliminaries we begin with the main part of the argument. Let $f\in\mathcal H$ with $\vvvert f \vvvert =1$. We define $u_f$ by \eqref{eq:uf} and set
	$$
	s_0 := \|u_f\|_{L^\infty(\Omega)}^2 \,.
	$$
	Then the quantity we are interested in can be written as
	$$
	\int_\Omega \Phi(u_f(z)^2) \,dm(z) = \int_0^{s_0} m(\{ u_f^2>s\}) \Phi'(s)\,ds \,.
	$$
	Here $\Phi'$ denotes either the left or the right-sided derivative of $\Phi$, which are known to exist everywhere and to coincide outside of a countable set \cite[Theorem 1.26]{Si3}. We also used the facts that $\Phi$ is absolutely continuous \cite[Theorem 1.28]{Si3} and that $\Phi(0)=0$.
	
	We now write the quantity on the right side as
	\begin{align*}
		\int_0^{s_0} \left( -\ln s - g(s^\frac12) \right) \Phi'(s)\,ds & \qquad \text{in Case 1}\,,\\
		\int_0^{s_0} \left( 1 - s^\frac1{2J} g(s^\frac12) \right) \Phi'(s)\,ds & \qquad \text{in Case 2}\,,\\
		\int_0^{s_0} \left( -1 - s^\frac1\alpha g(s^\frac12)\right) \Phi'(s)\,ds & \qquad \text{in Case 3}\,,
	\end{align*}
	where, according to Lemma \ref{mono}, $\kappa\mapsto g(\kappa)$ is nondecreasing on $(0,s_0^\frac12)$. In particular, when $\Phi$ is the identity, we obtain, in view of the normalization of $f$,
	\begin{equation}
		\label{eq:norm}
		1 = \vvvert f \vvvert^2 =
		\begin{cases}
			\int_0^{s_0} \left( -\ln s - g(s^\frac12) \right) ds & \qquad \text{in Case 1}\,,\\
			(2J+1) \int_0^{s_0} \left( 1 - s^\frac1{2J} g(s^\frac12) \right) ds & \qquad \text{in Case 2}\,,\\
			(\alpha-1) \int_0^{s_0} \left( -1 - s^{-\frac1\alpha} g(s^\frac12)\right) ds & \qquad \text{in Case 3}\,.
		\end{cases}
	\end{equation}
	Let us set
	\begin{align*}
		t_0 :=
		\begin{cases}
			s_0 & \qquad\text{in Case 1}\,, \\
			s_0^\frac{2J+1}{2J} & \qquad\text{in Case 2}\,,\\
			s_0^\frac{\alpha-1}{\alpha} & \qquad\text{in Case 3} \,,
		\end{cases}
	\end{align*}
	and, for $0\leq t\leq t_0$,
	\begin{equation*}
		h(t):= 
		\begin{cases}
			g(t^\frac12) & \qquad\text{in Case 1} \,,\\
			g(t^\frac{j}{2J+1}) & \qquad\text{in Case 2}\,,\\
			g(t^\frac{\alpha}{2(\alpha-1)}) & \qquad\text{in Case 3}\,.
		\end{cases}
	\end{equation*}
	Then the normalization \eqref{eq:norm} can be equivalently written as
	$$
	1 =
	\begin{cases}
		\int_0^{t_0} \left( -\ln t - h(t) \right) dt & \qquad \text{in Case 1}\,,\\
		2J \int_0^{t_0} \left( t^{-\frac1{2J+1}} - h(t) \right) dt & \qquad \text{in Case 2}\,,\\
		\alpha \int_0^{t_0} \left( - t^{\frac1{\alpha-1}} - h(t)\right) dt & \qquad \text{in Case 3}\,,
	\end{cases}
	$$
	while the quantity to be maximized is 
	\begin{equation}
		\label{eq:analyticproof}
		\int_\Omega \Phi(u_f(z)^2) \,dm(z) =
		\begin{cases}
			\int_0^{t_0} \left( -\ln t - h(t) \right) w(t)\,dt & \qquad \text{in Case 1}\,,\\
			2J \int_0^{t_0} \left( t^{-\frac1{2J+1}} - h(t) \right) w(t)\,dt & \qquad \text{in Case 2}\,,\\
			\alpha \int_0^{t_0} \left( - t^{\frac1{\alpha-1}} - h(t)\right) w(t)\,dt & \qquad \text{in Case 3}\,,
		\end{cases}
	\end{equation}
	where, for $0\leq t\leq t_0$,
	$$
	w(t) :=
	\begin{cases}
		\Phi'(t) & \qquad \text{in Case 1}\,,\\
		\Phi'(t^\frac{2J}{2J+1}) & \qquad \text{in Case 2}\,,\\
		\Phi'(t^\frac\alpha{\alpha-1}) & \qquad \text{in Case 3}\,.
	\end{cases}
	$$
	Since $\Phi$ is convex, $\Phi'$ is nondecreasing and therefore $w$ is nondecreasing as well. Also, $h$ is nondecreasing since $g$ is. Thus, Lemma \ref{kulikov} is applicable and, for given $t_0$, an upper bound on the right side of \eqref{eq:analyticproof} is obtained by replacing $h$ by $t_0^{-1} \int_0^{t_0} h(t)\,dt$. According to the normalization, we have
	$$
	t_0^{-1} \int_0^{t_0} h(t)\,dt = C(t_0) \,,
	$$
	where, for $\tau\in(0,1]$,
	$$
	C(\tau) := \begin{cases}
		-\tau^{-1} - \tau^{-1} \int_0^{\tau}\ln t\,dt & \qquad \text{in Case 1}\,,\\
		-(2J)^{-1}\tau^{-1} + \tau^{-1} \int_0^{\tau} t^{-\frac1{2J+1}}\,dt & \qquad \text{in Case 2}\,,\\
		-\alpha^{-1} \tau^{-1} - \tau^{-1} \int_0^{\tau} t^{\frac1{\alpha-1}} \,dt & \qquad \text{in Case 3}\,.
	\end{cases}
	$$
	Thus, we have shown the upper bound
	\begin{equation}
		\label{eq:upperboundproof}
		\int_\Omega \Phi(u_f(z)^2) \,dm(z) \leq A(t_0)
	\end{equation}
	where, for $\tau\in(0,1]$,
	$$
	A(\tau) :=
	\begin{cases}
		\int_0^{\tau} \left( -\ln t - C(\tau) \right) w(t)\,dt & \qquad \text{in Case 1}\,,\\
		2J \int_0^{\tau} \left( t^{-\frac1{2J+1}} - C(\tau) \right) w(t)\,dt & \qquad \text{in Case 2}\,,\\
		\alpha \int_0^{\tau} \left( - t^{\frac1{\alpha-1}} - C(\tau) \right) w(t)\,dt & \qquad \text{in Case 3}\,.
	\end{cases}
	$$
	
	Our goal now is to show that $A$ is nondecreasing in $(0,1]$. Since, by Lemma \ref{linfty}, $t_0\leq 1$, inserting this into \eqref{eq:upperboundproof} gives us the upper bound $A(1)$. Later, we will argue that this is the claimed optimal bound and discuss the cases of equality. 
	
	In order to prove the monotonicity of $A$, we first compute
	$$
	C(\tau) =
	\begin{cases}
		-\tau^{-1} - \ln \tau + 1 & \qquad \text{in Case 1}\,,\\
		-(2J)^{-1} \tau^{-1} + \frac{2J+1}{2J} \tau^{-\frac{1}{2J+1}}& \qquad \text{in Case 2}\,,\\
		-\alpha^{-1} \tau^{-1} -\frac{\alpha-1}{\alpha} \tau^{\frac1{\alpha-1}}  & \qquad \text{in Case 3}\,.
	\end{cases}
	$$
	From these expressions one easily deduces that $C'>0$ in $(0,1)$. Another consequence that we will use soon is that
	\begin{equation}
		\label{eq:ceq}
		0 = 
		\begin{cases}
			-\ln \tau - C(\tau) - C'(\tau) \tau & \qquad \text{in Case 1}\,,\\
			\tau^{-\frac1{2J+1}} - C(\tau) - C'(\tau) \tau & \qquad \text{in Case 2}\,,\\
			- \tau^{\frac1{\alpha-1}} - C(\tau) - C'(\tau) \tau & \qquad \text{in Case 3}\,.
		\end{cases}
	\end{equation}
	We now compute
	$$
	A'(\tau) = 
	\begin{cases}
		\left( -\ln \tau - C(\tau) \right) w(\tau) - C'(\tau) \int_0^{\tau} w(t)\,dt & \qquad \text{in Case 1}\,,\\
		2J \left( \tau^{-\frac1{2J+1}} - C(\tau) \right) w(\tau) - 2J C'(\tau) \int_0^{\tau} w(t)\,dt & \qquad \text{in Case 2}\,,\\
		\alpha \left( - \tau^{\frac1{\alpha-1}} - C(\tau) \right) w(\tau) - \alpha C'(\tau) \int_0^{\tau} w(t)\,dt & \qquad \text{in Case 3}\,.
	\end{cases}
	$$
	Since $w$ is increasing, we have $\int_0^{\tau}w(t)\,dt \leq \tau w(\tau)$. This, together with $C'(\tau)\geq 0$, implies
	$$
	A'(\tau) \geq 
	\begin{cases}
		\left( -\ln \tau - C(\tau) \right) w(\tau) - C'(\tau) \tau w(\tau) & \qquad \text{in Case 1}\,,\\
		2J \left( \tau^{-\frac1{2J+1}} - C(\tau) \right) w(\tau) - 2J C'(\tau) \tau w(\tau) & \qquad \text{in Case 2}\,,\\
		\alpha \left( - \tau^{\frac1{\alpha-1}} - C(\tau) \right) w(\tau) - \alpha C'(\tau) \tau w(\tau) & \qquad \text{in Case 3}\,.
	\end{cases}
	$$
	According to \eqref{eq:ceq}, the right side is equal to zero in all cases. This proves that $A'\geq 0$ in $(0,1]$.
	
	As mentioned before, this allows us to replace $A(t_0)$ by $A(1)$ in \eqref{eq:upperboundproof}. We claim that this bound is optimal. Indeed, if $f\in\mathcal M$, then, by the second part of Lemma \ref{mono}, $g$ is constant. Thus, also $h$ is constant and nothing was lost when applying Lemma \ref{kulikov}. This proves that in this case \eqref{eq:upperboundproof} is an equality and, since $t_0=1$ by Lemma \ref{linfty}, we have shown the claimed optimality.
	
	Finally, assume that $A(1)<\infty$ and that $\Phi$ is not affine linear. Then $\Phi'$ is not constant and neither is $w$. We deduce that there is an $\epsilon>0$ such that the inequality $\int_0^\tau w(t)\,dt \leq \tau w(\tau)$ is strict for all $\tau\in(1-\epsilon,1]$. This, together with the fact that $C'(\tau)>0$ for all $\tau\in(0,1)$ implies that $A'(\tau)>0$ for all $\tau\in(1-\epsilon,1)$. In particular, $A(\tau)<A(1)$ if $\tau\in[0,1)$.
	
	As a consequence, if $f\in\mathcal H$ with $\vvvert f\vvvert=1$ attains the supremum in Theorem \ref{analytic}, then necessarily $t_0=1$. Then, by Lemma \ref{linfty}, $f\in\mathcal M_1$, as claimed. This completes the proof of Theorem \ref{analytic}, except for the explicit value of the supremum.
	
	To compute the latter, we may choose an arbitrary element in $\mathcal M$ and it is convenient to take $f=F_0=1$. Then we obtain, by integrating in radial coordinates,
	$$
	\int_\Omega \Phi( u_{F_0}(z)^2)\,dz =
	\begin{cases}
		2\pi \int_0^\infty \Phi(e^{-\pi r^2})\,r\,dr & \qquad\text{in Case 1} \,,\\
		2\int_0^\infty \Phi((1+r^2)^{-2J})(1+r^2)^{-2} r\,dr & \qquad\text{in Case 2} \,,\\
		2\int_0^1 \Phi((1-r^2)^\alpha)(1-r^2)^{-2} r\,dr & \qquad\text{in Case 3} \,.
	\end{cases}
	$$
	Changing variables $s=e^{-\pi r^2}$, $s=(1+r^2)^{-2J}$ and $s=(1+r^2)^\alpha$ in the three cases we easily arrive at the claimed formulas.	
\end{proof}


\subsection{Extension to density matrices}

In this subsection we generalize the inequality in Theorem \ref{analytic} and, under a slightly stronger assumption on $\Phi$, we characterize the cases of equality. We use an argument similar to \cite[Lemma 2]{Li}.

Given an operator $\rho\geq 0$ with $\Tr\rho=1$ on one of the Hilbert spaces $\mathcal H$, we define a function $u_\rho$ on $\Omega$ as follows. We can write
$$
\rho = \sum_n p_n |f_n\rangle\langle f_n|
\qquad\text{with}\ \sum_n p_n = 1 \,,\ p_n\geq 0 \,,\ \langle f_n,f_m \rangle = \delta_{n,m} \,.
$$
We then set
$$
u_\rho(z) := \left( \sum_n p_n u_{f_n}(z)^2 \right)^\frac12 \,.
$$
It is easily checked that this is well-defined. (Note, in particular, the nonuniqueness of the above decomposition of $\rho$ in the case of a degenerate eigenvalue.) Moreover, for $\rho=|f\rangle\langle f|$, this definition of $u_\rho$ coincides with the earlier one of $u_f$.

\begin{corollary}\label{analyticcor}
	Let $\Phi:[0,1]\to\R$ be convex. Then
	\begin{align*}
		& \sup\left\{ \int_\Omega \Phi( u_\rho(z)^2) \,dm(z) :\ \rho\geq 0\ \text{on}\ \mathcal H \,,\ \Tr\rho =1 \right\} \\
		& =\begin{cases}
			\int_0^1 \Phi(s) s^{-1}\,ds & \qquad\text{in Case 1}\,,\\
			(2J)^{-1} \int_0^1 \Phi(s) s^{\frac{1}{2J}-1}\,ds & \qquad\text{in Case 2}\,,\\
			(2K)^{-1} \int_0^1 \Phi(s) s^{-\frac{1}{2K}-1}\,ds & \qquad\text{in Case 3}\,,
		\end{cases}
	\end{align*}
	and the supremum is attained for $\rho=|F\rangle\langle F|$ with $F\in \mathcal M$. If, moreover, $\Phi$ is strictly convex and if the supremum is finite, then it is attained \emph{only} for such $\rho$.
\end{corollary}

\begin{proof}
	We use the above expansion of $\rho$. By convexity of $\Phi$, for any $z\in\Omega$,
	$$
	\Phi( u_\rho(z)^2) = \Phi( \sum_n p_n u_{f_n}(z)^2 ) \leq \sum_n p_n \Phi( u_{f_n}(z)^2 ) \,.
	$$
	Thus, with $S$ denoting the supremum in Theorem \ref{analytic},
	$$
	\int_\Omega \Phi( u_\rho(z)^2) \,dm(z) \leq \sum_n p_n \int_\Omega \Phi( u_{f_n}(z)^2 ) \,dm(z) \leq \sum_n p_n S = S \,.
	$$
	Since, by Theorem \ref{analytic}, $S$ is attained for $\rho=|F\rangle\langle F|$ with $F\in \mathcal M$, we obtain the first assertion in the corollary.
	
	Now assume that $S<\infty$ and that equality is achieved for some $\rho$. If $\Phi$ is not linear, then, by Theorem \ref{analytic}, $f_n\in\mathcal M$ for each $n$. (Throughout we restrict ourselves to values of $n$ for which $p_n>0$.) Moreover,
	$$
	\Phi( \sum_n p_n u_{f_n}(z)^2 ) = \sum_n p_n \Phi( u_{f_n}(z)^2 ) 
	\qquad\text{for a.e.}\ z\in\Omega \,.
	$$
	Assuming now that $\Phi$ is strictly convex, we deduce that $u_{f_n}(z)^2 = u_{f_1}(z)^2$ for a.e.\ $z\in\Omega$ and every $n$. Thus, by continuity, $|f_n(z)|=|f_1(z)|$ for all $z\in\Omega$ and all $n$. By analyticity, there are $\theta_n\in\R/2\pi\Z$ such that $f_n(z) = e^{i\theta_n} f_1(z)$ for all $z\in\Omega$. (Indeed, by the maximum modulus principle $f_n/f_1$ is equal to a constant in $\Omega$ without the zeros of $f_1$ and then by continuity in all of $\Omega$.) Since $\langle f_n,f_1\rangle=\delta_{n,1}$, we conclude that there is only a single index $n$, namely, $n=1$.	
\end{proof}


\subsection{Another inequality of Kulikov}

For later purposes, in this subsection we record another inequality from \cite{Ku} which corresponds, in some sense, to the limiting case $\alpha=1$ in Theorem \ref{analytic}.

The underlying Hilbert space is the Hardy space $H^2(\D)$ consisting of all analytic functions $f$ in $\D$ such that
$$
\| f \|_{H^2(\D)} := \left( \sup_{0<r<1} (2\pi)^{-1} \int_{-\pi}^\pi |f(re^{i\phi})|^2\,d\phi \right)^{1/2}<\infty \,.
$$
To emphasize the analogy with Theorem \ref{analytic} we denote the space by $\mathcal H$ and the norm by $\vvvert \cdot\vvvert$. We also use the same notation $\Omega$ and $dm(z)$ as in Case 3. The function $u_f$ is defined by \eqref{eq:uf} with $\alpha=1$. The functions $F_w$ are defined as in Case 3 with $\alpha=1$ and one easily checks that they are normalized. The set $\mathcal M$ is defined as before.

\begin{proposition}\label{analyticlimit}
	Let $\Phi:[0,1]\to\R$ be nondecreasing. Then
	$$
	\sup\left\{ \int_\Omega \Phi( u_f(z)^2) \,dm(z) :\ f\in\mathcal H \,,\ \vvvert f\vvvert =1 \right\} = \int_0^1 \Phi(s) s^{-2}\,ds
	$$
	and the supremum is attained in $\mathcal M$. If $\Phi$ is strictly increasing near 1 and if the supremum is finite, then it is attained \emph{only} in $\mathcal M$.
\end{proposition}

By $\Phi$ being `strictly increasing near 1' we mean that $\Phi(s)<\Phi(1^-)$ for all $s<1$, where $\Phi(1^-):=\lim_{s\to1^-}\Phi(s)$. We will see in the proof that $u_f\leq 1$, so the quantity in the supremum is well defined.

\begin{proof}
	The first part is a special case of \cite[Theorem 1.1]{Ku}. The second part can be obtained by an inspection of the proof of the first part, but, since this is not explicitly stated in \cite{Ku}, we provide some details. As in the proof of Theorem \ref{analytic}, we may assume that $\Phi(0)=\lim_{s\to 0^+}\Phi(s)=0$. Then
	$$
	\int_\Omega \Phi( u_f(z)^2) \,dm(z) = \int_0^{s_0} \mu(s^\frac12)\,d\Phi(s) \,,
	$$
	where $s_0:=\|u_f\|_{L^\infty(\Omega)}^2$. By an analogue of Lemma \ref{linfty}, we have $s_0\leq 1$ \cite[(1.2)]{Ku} with equality if and only if $f\in\mathcal M$. The argument for the latter assertion is essentially the same as in Lemma \ref{linfty}, using the fact that $u_f(z)\to 0$ as $|z|\to 1$, stated in \cite[paragraph after (1.2)]{Ku}, and the cases of a Schwarz inequality for a reproducing kernel (see also the proof of Proposition \ref{su11limit} below).
	
	As shown in \cite[Theorem 3.1]{Ku}, we have $\mu(\kappa)\leq (\kappa^{-2}-1)_+$ for all $\kappa>0$ and this bound is an equality if $f\in\mathcal M$. Thus,
	$$
	\int_0^{s_0} \mu(s^\frac12)\,d\Phi(s) \leq \int_0^1 (s^{-1}-1)\,d\Phi(s) - \int_{s_0}^1 (s^{-1}-1)\,d\Phi(s) \,,
	$$
	where the first term on the right side corresponds to the value of the supremum. Thus, if this term is finite and $f$ attains the supremum, then the second term on the right sides has to vanish. If $\Phi$ is strictly increasing near 1, then the measure $d\Phi(s)$ does not vanish on any interval $(1-\epsilon,1)$ with $\epsilon>0$ and therefore, necessarily, $s_0=1$. By the above, this means $f\in\mathcal M$, as claimed.
	
	To compute the value of the supremum we can proceed exactly as in Case 3 of Theorem \ref{analytic}, setting $\alpha=1$ in that calculation. This proves the proposition.	
\end{proof}


\section{Reverse H\"older inequalities for analytic functions}\label{sec:reverse}

The material in this section is an extension of that in the previous section. It is not relevant for the proof of the results in Section \ref{sec:intro}

In Theorem \ref{analytic} we were working under a constraint on a Hilbertian norm. It turns out that this is an unnecessary restriction. We will prove a generalization of Theorem~\ref{analytic} with a constraint on a more general norm or quasinorms. This will allow us to settle a conjecture by Bodmann \cite[Conjecture 3.5]{Bo}.

We continue to use the notation of Section \ref{sec:analytic}. For $0<p<\infty$, we define
$$
\vvvert f \vvvert_p :=
\begin{cases}
	\left( \frac p2 \int_\C |f(z)|^p e^{-p \frac{\pi}{2}|z|^2}\,dA(z) \right)^{1/p} & \qquad\text{in Case 1} \,,\\
	\left( \frac{pJ+1}{\pi} \int_\C |f(z)|^p (1+|z|^2)^{-pJ-2}\,dA(z) \right)^{1/p} & \qquad\text{in Case 2} \,,\\
	\left( \frac{\alpha p-2}{2\pi} \int_\D |f(z)|^p (1-|z|^2)^{p\frac\alpha 2-2}\,dA(z) \right)^{1/p} & \qquad\text{in Case 3} \,.
\end{cases}
$$
This is a norm for $p\geq 1$ and a quasinorm for $p<1$. The prefactors are chosen such that $\vvvert 1\vvvert_p =1$. We still assume that $J\in\frac12\N$ in Case 2 and now $\alpha>\frac2p$ in Case 3. The function $u_f$ is defined as before. We denote by $\mathcal X^p$ the space of all analytic functions $f$ on $\Omega$ such that $\vvvert f\vvvert_p<\infty$. In Case 2 we require, in addition, that $f$ is a polynomial of degree $\leq 2J$. We note that every $F\in\mathcal M$ satisfies $\vvvert F \vvvert_p=1$. Indeed, we have already noted that this holds for $F=F_0=1$ and for general $F\in\mathcal M$ it follows from the equimeasurability of $u_{F_w}$ for different $w$, discussed before Theorem \ref{analytic}.

\begin{theorem}\label{analyticp}
	Let $0<p<\infty$ and let $\Phi:[0,1]\to\R$ be convex. Then
	\begin{align*}
		& \sup\left\{ \int_\Omega \Phi( u_f(z)^p) \,dm(z) :\ f\in\mathcal X^p \,,\ \vvvert f\vvvert_p =1 \right\} \\
		& =
		\begin{cases}
			\frac2p \int_0^1 \Phi(s) s^{-1}\,ds & \qquad\text{in Case 1}\,,\\
			(pJ)^{-1} \int_0^1 \Phi(s) s^{\frac{1}{pJ}-1}\,ds & \qquad\text{in Case 2}\,,\\
			\frac2{\alpha p} \int_0^1 \Phi(s) s^{-\frac{2}{\alpha p}-1}\,ds & \qquad\text{in Case 3}\,,
		\end{cases}
	\end{align*}
	and the supremum is attained in $\mathcal M$. If $\Phi$ is not affine linear and if the supremum is finite, then it is attained \emph{only} in $\mathcal M$.
\end{theorem}

For $p=2$ this theorem reduces to Theorem \ref{analytic}. In Case 3 it reduces to \cite[Theorem 1.2 and Remark 4.3]{Ku}, except that our equality statement allows for more general $\Phi$. In Cases 1 and 2 the theorem seems to be new.

Taking $\Phi(s) = s^\frac{q}{p}$ with $q>p$ we obtain the following reverse H\"older inequalities.

\begin{corollary}\label{reverseh}
	Let $0<p<q<\infty$. Then, for any $f\in\mathcal X^p$,
	$$
	\vvvert f \vvvert_q \leq \vvvert f \vvvert_p
	$$
	with equality if and only if $f=0$ or $\vvvert f\vvvert_p^{-1} f\in\mathcal M$.
\end{corollary}

This corollary in Case 1 is due to Carlen \cite[Theorem 2]{Ca}. In fact, Carlen proves a more general inequality including an additional parameter. Carlen's method of proof depends on the logarithmic Sobolev inequality and an identity for analytic functions. It is different from ours. Corollary \ref{reverseh} in Case 2 has been conjectured by Bodmann \cite[Conjecture 3.5]{Bo}, who proved it in the special case where $q = p + J^{-1}n$ where $n\in\N$ and $p>J^{-1}$. Bodmann's proof relies on a sharp Sobolev inequality and an analogue of Carlen's identity. Corollary \ref{reverseh} in Case 3 is due to Kulikov \cite[Corollary 1.3]{Ku}. The special case $q=p+2 \alpha^{-1}$ with $p\geq 2$, $\alpha p>4$ was earlier proved by Bandyopadhyay in \cite[Corollary 3.3]{Ba} using the method of Carlen and Bodmann. (Note that in \cite{Ba} it is assumed that $\alpha\in\N\setminus\{1\}$ -- in her notation $\alpha=2k$ --, but this seems to be irrelevant for \cite[Section 3]{Ba}.)

We turn now to the proof of Theorem \ref{analyticp}. The main new ingredient is the following generalization of Lemma \ref{linfty}. In Case 3 this is well known \cite[(1.1)]{Ku} and probably also in Case 1, but in Case 2 it might be new.

\begin{lemma}\label{linftyp}
	Let $0<p<\infty$ and let $f\in\mathcal X^p$. Then
	$$
	\| u_f \|_{L^\infty(\Omega)} \leq \vvvert f \vvvert_p \,.
	$$
	with equality if and only if either $f=0$ or $\vvvert f\vvvert_p^{-1} f\in\mathcal M$.
\end{lemma}

\begin{proof}
	We begin by showing that
	\begin{equation}
		\label{eq:linftyp0}
		u_f(0) \leq \vvvert f \vvvert_p \,.
	\end{equation}
	with equality if and only if $\vvvert f\vvvert_p^{-1} f=e^{i\theta} F_0$ for some $\theta\in\R/2\pi\Z$ (provided $f\neq 0$). Since $\ln|f|$ is subharmonic in $\Omega$, we have, for any $r\in(0,R)$, where $R:=\infty$ in Cases~1 and 2 and $R:=1$ in Case 3,
	$$
	\ln |f(0)| \leq (2\pi)^{-1} \int_{-\pi}^\pi \ln|f(re^{i\phi})\,d\phi \,.
	$$
	We multiply by $r e^{-p\frac\pi2 r^2}$, $r(1+r^2)^{-pJ-2}$ and $r(1-r^2)^{p\frac\alpha2-2}$ in the different cases and integrate with respect to $r\in(0,R)$. In this way, we obtain
	$$
	\ln |f(0)| \leq \int_\Omega \ln |f(z)| w(z)\,dA(z) \,,
	$$
	where
	$$
	w(z) :=
	\begin{cases}
		\frac p2 e^{-p\frac\pi2|z|^2} & \qquad\text{in Case 1}\,,\\
		\frac{pJ+1}{\pi} (1+|z|^2)^{-pJ-2} & \qquad\text{in Case 2}\,,\\
		\frac{\alpha p-2}{2\pi} (1-|z|^2)^{p\frac{\alpha}{2}-2} & \qquad\text{in Case 3} \,. 
	\end{cases}
	$$
	The measure $w(z)\,dA(z)$ is a probability measure. Multiplying by $p$ we can write this as
	$$
	|f(0)|^p \leq \exp\left( \int_\Omega \ln (|f(z)|^p) w(z)\,dA(z) \right) \leq \int_\Omega |f(z)|^p w(z)\,dA(z) = \vvvert f \vvvert_p^p \,,
	$$
	where the second inequality comes from Jensen's inequality. Since the exponential function is strictly convex, Jensen's inequality is strict unless $\ln(|f|^p)$ is almost everywhere constant. Since $f$ is continuous, this happens if and only $|f|$ is constant and, by the maximum modulus principle if and only $f$ is constant. This proves the claim.
	
	We now claim that for any $z\in\Omega$,
	\begin{equation}
		\label{eq:linftypz}
		u_f(z) \leq \vvvert f\vvvert_p
	\end{equation}
	if and only if $\vvvert f\vvvert_p^{-1} f=e^{i\theta} F_z$ for some $\theta\in\R/2\pi\Z$ (provided $f\neq0$). In Case 2, the same inequality remains valid for $z=\infty$, recalling that $u_f$ extends continuously to this point. Indeed, inequality \eqref{eq:linftypz} and its equality statement follow from the corresponding assertions concerning \eqref{eq:linftyp0}, by applying an element of the Heisenberg group, SU(2) or SU(1,1) to move the point $z$ to the point $0$ and by noting that $\vvvert \cdot \vvvert_p$ is invariant under this group action. The latter fact follows from the equimeasurability property discussed before Theorem \ref{analytic}.
	
	Inequality \eqref{eq:linftypz} implies the inequality in the lemma. Now assume that $f\neq 0$ achieves equality in this inequality. We claim that $u_f(z)\to 0$ as $|z|\to\infty$ or $|z|\to 1$ in Cases 1 and 3. This, together with the continuity of $u_f$ in $\Omega$ in Cases 1 and 3 and in $\C\cup\{\infty\}$ in Case 2, implies that there is a $z$ such that $u_f(z)=\|u_f\|_{L^\infty(\Omega)}$. The equality statement in \eqref{eq:linftypz} then implies the equality statement in the lemma.
	
	Thus, it suffices to prove the asymptotic vanishing of $u_f$ in Cases 1 and 3. In both cases, this is clear when $f$ is a polynomial and follows in the general case from the fact that polynomials are dense with respect to $\vvvert\cdot\vvvert_p$ and the inequality in the lemma. This completes the proof.	
\end{proof}

\begin{proof}[Proof of Theorem \ref{analyticp}]
	Given Lemma \ref{linftyp}, which replaces Lemma \ref{linfty}, the proof is a minor variation of that of Theorem \ref{analytic}. We only sketch the major steps. The task is to maximize
	$$
	\int_0^{s_0} m(\{ u_f^p>s\}) \Phi'(s)\,ds = \int_0^{s_0} \mu(s^\frac1p) \Phi'(s)\,ds
	$$
	under the constraint
	$$
	\int_0^{s_0} m(\{ u_f^p>s\}) \,ds = \int_0^{s_0} \mu(s^\frac1p) \,ds =
	\begin{cases}
		\frac2p \,, \\ \frac{1}{pJ+1} \,, \\ \frac{2}{\alpha p-2} \,,
	\end{cases} 
	$$
	with $s_0:= \|u_f\|_{L^\infty(\Omega)}^p$ and $\mu$ as in Lemma \ref{mono}. The latter lemma allows us to write $\mu(s^\frac1p)$ as the sum of a fixed piece and one that involves the nondecreasing function $g(s^\frac1p)$. We pass from the variable $s$ to a variable $t$ so that the resulting nondecreasing function $h$ is $L^1$-normalized with respect to the unweighted measure $dt$. Then we can use Chebyshev's bound (Lemma \ref{kulikov}) to replace $h$ by its average. This leads to a certain bound $A(t_0)$ and a computation, similarly as for $p=2$, shows that $A$ is nondecreasing. Moreover, if $\Phi$ is not affine linear, then $A$ is strictly increasing. This concludes the sketch of the proof of Theorem \ref{analyticp}.	
\end{proof}


\section{Proof of the main results}\label{sec:proofs}

In this section we prove the main results stated in the introduction. In each case we will work with a concrete representation of the group action that involves analytic functions. The inequalities will then be deduced from Theorem \ref{analytic}.


\subsection{Proof of Theorem \ref{heisen}}

By scaling, it suffices to prove the theorem for a single value of $\hbar$ and it is convenient to choose $\hbar=(2\pi)^{-1}$. Then, given $\psi\in L^2(\R)$, we can write
$$
\langle \psi_{p,q},\psi \rangle = e^{-\frac{\pi}{2}(q^2+p^2)} e^{i\pi qp} f(q-ip)
$$
with
$$
f(z) := 2^\frac14 \int_\R e^{2\pi z x - \tfrac\pi2z^2 - \pi x^2} \psi(x)\,dx \,.
$$
It is well known and easy to see that $f$ is entire and that
$$
\| f \|_{\mathcal F^2}^2 = \iint_{\R\times\R} |\langle \psi_{p,q},\psi \rangle|^2\,dp\,dq = \|\psi\|_{L^2(\R)}^2 \,,
$$
where the last identity is the completeness relation of the coherent states. In particular, $f\in\mathcal F^2(\C)$. Moreover,
$$
u_f(q-ip) = | \langle \psi_{p,q},\psi \rangle | \,,
$$
so Theorem \ref{heisen} follows immediately from Theorem \ref{analytic} in Case 1. Similarly, Corollary \ref{heisencor} follows from Corollary \ref{analyticcor}.


\subsection{Proof of Theorem \ref{su2}}

Let $J\in\frac12\N$. We consider the representation of SU(2) on functions $f$ on $\C$ given by
$$
\pi_U(f)(z) := (\beta z + \overline\alpha)^{2J}\ f(\tfrac{\alpha z - \overline\beta}{\beta z+\overline\alpha})
\qquad\text{for all}\ z\in\C \,, 
$$
where
$$
U = 
\begin{pmatrix}
	\alpha & \beta \\ -\overline\beta & \overline\alpha
\end{pmatrix}
\in \mathrm{SU(2)} \,,
\qquad\text{that is},\qquad \alpha,\beta\in\C \ \text{with}\
|\alpha|^2 + |\beta|^2 = 1 \,.
$$
This representation restricted to $\mathcal P_{2J}$ is irreducible and unitary for the norm defined above. In this representation,
$$
S_1 = \frac12\left( (-z^2+1) \frac{d}{dz} +2Jz \right),
\quad
S_2 = \frac1{2i}\left( (-z^2-1) \frac{d}{dz} +2Jz \right),
\quad
S_3 = z\frac{d}{dz}- J \,.
$$

We may choose the space $\mathcal H$ in Theorem \ref{su2} as $\mathcal P_{2J}$. By an explicit computation one sees that the functions $F_w$ are eigenvectors of the operator $\mathcal S(w)\cdot S$ corresponding to the eigenvalue $-J$, where we used the stereographic projection $\mathcal S:\C\mapsto \Sph^2$, given by
$$
\mathcal S_1(w)+i\mathcal S_2(w) := \frac{2w}{1+|w|^2} \,,
\qquad
\mathcal S_3(w) := \frac{1-|w|^2}{1+|w|^2} \,.
$$
Consequently, the phases of the $\psi_\omega$, $\omega\in\Sph^2$, can be chosen such that these functions coincide with the functions $F_w$, $w\in\C\cup\{\infty\}$. Thus, using the explicit form of the $F_w$,
$$
\langle\psi_{\mathcal S(w)},\psi\rangle = (1+|w|^2)^{-J} f(w)
$$
with
$$
f(w) := \frac{2J+1}{\pi} \int_\C (1+w\overline z)^{2J} \psi(z) (1+|z|^2)^{-2J-2}\,dA(z) \,.
$$
Since $\psi$ is a polynomial of degree $\leq 2J$, the reproducing property of the kernel implies that $f(w)=\psi(w)$ for all $w\in\C$. Moreover,
$$
u_f(w) = |\langle\psi_{\mathcal S(w)},\psi\rangle|
$$
and so, by a change of variables,
$$
\pi^{-1} \int_\C \Phi(u_f(w)^2)\,\frac{dA(w)}{(1+|w|^2)^2} = (4\pi)^{-1} \int_{\Sph^2} \Phi(|\langle\psi_\omega,\psi\rangle|^2)\,d\omega \,.
$$
Thus, Theorem \ref{su2} follows immediately from Theorem \ref{analytic} in Case 2. Similarly, Corollary~\ref{su2cor} follows from Corollary \ref{analyticcor}.


\subsection{Proof of Theorem \ref{su11}}

Let $K\in\frac12\N\setminus\{\frac12\}$. We consider the representation of SU(1,1) on functions $f$ on $\D$ given by
$$
\pi_U(f)(z) := (\beta z + \overline\alpha)^{-2K}\ f(\tfrac{\alpha z + \overline\beta}{\beta z+\overline\alpha})
\qquad\text{for all}\ z\in\C \,, 
$$
where
$$
U = 
\begin{pmatrix}
	\alpha & \beta \\ \overline\beta & \overline\alpha
\end{pmatrix}
\in \mathrm{SU(1,1)} \,,
\qquad\text{that is},\qquad \alpha,\beta\in\C \ \text{with}\
|\alpha|^2 - |\beta|^2 = 1 \,.
$$
This representation restricted to $A_{2K}^2(\D)$ is irreducible and unitary for the norm defined above. In this representation,
$$
K_0 = z\frac{d}{dz}+K \,,
\quad
K_1 = \frac1{2i}\left( (z^2-1) \frac{d}{dz} +2Kz \right),
\quad
K_2 = -\frac1{2}\left( (z^2+1) \frac{d}{dz} +2Jz \right).
$$

We may choose the space $\mathcal H$ in Theorem \ref{su11} as $A_{2K}^2(\D)$. By an explicit computation one sees that the functions $F_w$ are eigenvectors of the operator $n_0K_0-n_1K_1-n_2K_2$ corresponding to the eigenvalue $K$. Here $(n_0,n_1,n_2)$ is related to $w$ as in the discussion before the statement of Theorem \ref{su11}. Consequently, we can choose the phases in such a way that $\psi_w=F_w$ for all $w\in\D$. Thus, using the explicit form of the $F_w$,
$$
\langle\psi_w,\psi\rangle = (1+|w|^2)^K f(w)
$$
with
$$
f(w) := \frac{2K-1}{\pi} \int_\D (1+w\overline z)^{-2K} \psi(z) (1+|z|^2)^{2K-2}\,dA(z) \,.
$$
Since $\psi\in A^2_{2K}(\D)$, the reproducing property of the kernel implies that $f(w)=\psi(w)$ for all $w\in\D$. Moreover,
$$
u_f(w) = |\langle\psi_w,\psi\rangle| \,,
$$
so Theorem \ref{su11} follows immediately from Theorem \ref{analytic} in Case 3. Similarly, Corollary \ref{su11cor} follows from Corollary \ref{analyticcor}.


\subsection{The limit of the discrete series}\label{sec:limit}

Two other irreducible unitary representations of SU(1,1) are not in the discrete series, but  are closely related to it, the so-called limits of discrete series \cite[Chapter II]{Kn}. They are typically not considered in the context of coherent states, since they are not square-integrable, but the questions discussed in this paper make perfectly sense for them and can be completely answered.

We restrict our attention to one of the limits of the discrete series, since the results for the other one can be deduced by appropriate complex conjugation. The construction of the coherent states is verbatim the same as for the discrete series, except that the value of $K$ now is $\frac12$.

\begin{proposition}\label{su11limit}
	Consider the irreducible limit of the discrete series representation of $\mathrm{SU(1,1)}$ on $\mathcal H$. Let $\Phi:[0,1]\to\R$ be nondecreasing. Then
	$$
	\sup\left\{ \int_{\D} \Phi(|\langle \psi_z, \psi \rangle|^2) \,\frac{dA(z)}{(1-|z|^2)^2} :\ \psi\in\mathcal H \,,\ \|\psi\|_\mathcal H= 1 \right\} = \pi \int_0^1 \Phi(s) s^{-2}\,ds
	$$
	and the supremum is attained for $\psi=e^{i\theta}\psi_{z_0}$ with some $z_0\in\D$, $\theta\in\R/2\pi\Z$. If $\Phi$ is strictly increasing near 1 and if the supremum is finite, then it is attained \emph{only} for such $\psi$.
\end{proposition}

Note that the value of the integral with $\psi=e^{i\theta}\psi_{z_0}$ does not depend on $z_0\in\D$, $\theta\in\R/2\pi\Z$. It may or may not be finite, depending on $\Phi$. For finiteness it is necessary that $\lim_{s\to 0^+}s^{-1}\Phi(s) = 0$. In particular, the function $\Phi(s)=s$ leads to an infinite supremum, which reflects the non-squareintegrability of the representation.

\begin{proof}[Proof of Proposition \ref{su11limit}]
	We consider the same representation of SU(1,1) on functions on $\D$ as in the proof of Theorem \ref{su11} but with $K=\frac12$. This representation is irreducible when restricted to the Hardy space $H^2(\D)$ and unitary for the norm defined above; see \cite[Section II.6]{Kn}. We choose the representation space $\mathcal H=H^2(\D)$. The functions $F_w$ were defined before Proposition \ref{analyticlimit} and one verifies that, by an appropriate choice of phases, $\psi_w=F_w$. It is well known that functions in the Hardy space have radial boundary values in $L^2(\partial\D)$ and that in their norm it suffices to consider this boundary value. Thus, using the explict form of the $F_w$,
	$$
	\langle \psi_w,\psi \rangle = (2\pi)^{-1} \int_{-\pi}^\pi \overline{F_w(e^{i\phi})}\psi(e^{i\phi})\,d\phi = (1+|w|^2)^\frac12 f(w)
	$$
	with
	$$
	f(w):= (2\pi)^{-1} \int_{-\pi}^\pi (1-w e^{-i\phi})^{-1} \psi(e^{i\phi})\,d\phi \,.
	$$
	By the reproducing property of the kernel (seen, for instance, by expanding both functions in the integrand into a Fourier series), we see that $f(w)=\psi(w)$ for all $w\in\D$. Moreover,
	$$
	u_f(w) = |\langle \psi_w,\psi \rangle| \,,
	$$
	so Proposition \ref{su11limit} follows from Proposition \ref{analyticlimit}.	
\end{proof}

There is also an analogue of Corollary \ref{su11cor} extending Proposition \ref{su11limit} (with convex $\Phi$) to density matrices, but we omit it for the sake of brevity.


\subsection{Proof of Theorem \ref{axb}}

Given $\psi\in L^2(\R_+)$, we can write
$$
\langle\psi_{a,b},\psi\rangle = a^\beta f(ia-b)
$$
with
$$
f(z) := 2^\beta \Gamma(2\beta)^{-\frac12} \int_0^\infty x^{\beta-\frac12} e^{izx} \psi(x)\,dx \,.
$$
It is easy to see and known that $f$ is analytic in $\C_+=\{ z\in\C:\ \im z>0\}$ and that
$$
\frac{\beta-\frac12}{2\pi} \int_{\C_+} |f(z)|^2 (\im z)^{2\beta-2} \,dA(z) = \frac{\beta-\frac12}{2\pi} \iint_{\R_+\times\R} | \langle\psi_{a,b},\psi\rangle |^2\, \frac{da\,db}{a^2} = \|\psi\|^2_{L^2(\R_+)} \,,
$$
where the last identity is the completeness relation of the coherent states \cite[(2.10)]{DaKlPa}. Consider the conformal map $\Sigma:\C_+\to\D$,
$$
\Sigma(z) = \frac{z-i}{-iz+1} \quad\text{for}\ z\in\C_+\,,
\qquad
\Sigma^{-1}(\zeta) = \frac{\zeta+i}{i\zeta+1} \quad \text{for}\ \zeta\in\D \,.
$$
Setting
$$
g(\zeta) := (i\zeta+1)^{-2\beta} f(\tfrac{\zeta+i}{i\zeta+1})
\qquad\text{for all}\ \zeta\in\D \,,
$$
we find that $g$ is analytic in $\D$ and, using $dA(\zeta) = |\Sigma'(z)|^2\,dA(z)$ for $\zeta=\Sigma(z)$, 
$$
\frac{\beta-\frac12}{2\pi} \int_{\C_+} |f(z)|^2 (\im z)^{2\beta-2} \,dA(z)
= \frac{2\beta-1}{\pi} \int_\D |g(\zeta)|^2 (1-|\zeta|^2)^{2\beta-2}\,dA(\zeta) = \| g \|_{A_{2\beta}(\D)}^2 \,.
$$
Moreover, after a simple computation,
$$
u_g(\Sigma(ia-b)) = |\langle\psi_{a,b},\psi\rangle|
$$
and therefore
$$
\iint_{\R_+\times\R} \Phi(|\langle\psi_{a,b},\psi\rangle|^2 )\,\frac{da\,db}{a^2} = 4\pi \int_\D \Phi(u_g(\zeta)^2) \,dm(\zeta) \,.
$$
Also, the coherent states $\psi_{a,b}$, $(a,b)\in\R_+\times\R$, are in one-to-one correspondence with the functions $F_w$, $w\in\D$. Indeed, a straightforward computation shows that the $f$ corresponding to $\psi=\psi_{1,0}$ is $f(z)=(2i/(z+i))^{2\beta}$, which corresponds to $g(\zeta)=1=F_0(\zeta)$. The result in the general case follows from the facts that every $(a,b)$ can be moved  to $(1,0)$ by an $aX+b$-action, every point $w\in\D$ can be moved to $0$ by the action of a subgroup of SU(1,1) isomorphic to $aX+b$ (see, e.g., \cite[equation after (1.3)]{LiSo3}) and that $\Sigma$ relates these actions to each other.

Thus, Theorem \ref{axb} follows immediately from Theorem \ref{analytic} in Case 3 with $\alpha=2\beta$ and, similarly Corollary \ref{axbcor} follows from Corollary \ref{analyticcor}.

\begin{remark}\label{axblimit}
	The functions $\psi_{a,b}$ are also well defined for $\beta\in(0,\frac12]$. In this case, the coherent state transform cannot be normalized to be an isometry to a subset of $L^2(\R_+\times\R,a^{-2}da\,db)$, but the optimization problem in Theorem \ref{axb} still makes sense. We claim that, for $\beta=\frac12$, Theorem \ref{axb} remains valid, replacing the assumptions `convex' and `not linear' on $\Phi$ by `nondecreasing' and `strictly increasing near 1', respectively. Indeed, in this case the function $f$ in the previous proof belongs to the Hardy space $H^2(\C_+)$ and, by Plancherel, its norm in that space is equal to $\|\psi\|_{L^2(\R_+)}$. Mapping $\C_+$ to $\D$ via $\Sigma$, we can deduce the assertion from Proposition \ref{analyticlimit}. We do not know whether Theorem \ref{axb} extends to $\beta\in(0,\frac12)$.
\end{remark}


\subsection{Limitations of the method}\label{sec:limitation}

In this paper we have discussed the cases of the Heisenberg group, SU(2), SU(1,1) and the affine group. It is a natural question, potentially of relevance for representation theory, to which extent the results can be generalized to arbitrary Lie groups.

While the method of the present paper is able to treat various cases in a unified way, it will probably not be able to deal with the general case, as we argue now. One of the key ingredients in the argument is Lemma \ref{mono}, whose proof uses the fact that the superlevel sets of the overlap of two coherent states are isoperimetric set. This may fail in general.

Following Lieb and Solovej \cite{LiSo2}, we consider the case of symmetric representations of SU(N). We fix $N\geq 3$. The relevant representations are labeled by $M\in\N$ and we choose the representation space $\mathcal H$ to be the symmetric subspace of the tensor product $\otimes^M \C^N$. Coherent states are defined through elements of the form $\otimes^M z$. Note that if two $z$'s differ by a phase, then the vectors in $\mathcal H$ also differ by a phase and correspond to the same state. Thus, we will label the coherent states by points $z$ in the complex projective space
$$
\C P^{N-1} = \left\{ z\in\C^N :\ |z|=1 \right\}/\sim
$$
where $z\sim w$ if $z=e^{i\theta} w$ for some $\theta\in\R/2\pi\Z$. We denote integration with respect to the natural SU(N)-invariant probability measure on $\C P^{N-1}$ by $dz$.

Lieb and Solovej have solved the corresponding problem and shown that, for any convex $\Phi:[0,1]\to\R$,
$$
\sup\left\{ \int_{\C P^{N-1}} \Phi(|\langle \otimes^M z, \psi \rangle|^2)\,dz :\ \psi\in\mathcal H \,,\ \|\psi\|_{\mathcal H} = 1 \right\}
$$
is attained for coherent states.

If we tried to reprove this through the method in the present paper, we would consider the measure of the superlevel sets of the function $z\mapsto|\langle \otimes^M z, \psi \rangle|$ and try to prove some monotonicity properties of it. This monotonicity property should be saturated if $\psi$ is of the form $\otimes^M z_0$. In this special case, the level sets are of the form
\begin{equation*}
	\left\{ z\in\C P^{N-1}:\ | z^* z_0|^M>\kappa \right\}
\end{equation*}
These are geodesic balls (see, e.g., \cite[Example 2.110]{GaHuLa}) and, if we want to use the method based on an isoperimetric inequality, they should be optimizers for the isoperimetric inequality. (More precisely, this should hold for all $\kappa$ for which their measure is $\leq\frac12$; for $\kappa$ such that their measure is $\geq\frac12$ their complements should be optimizers.) This, however, is not the case for all $\kappa$, at least not for $N=4$, as pointed out in \cite[Appendix]{BadCEs}; see also \cite[Remark 4.2]{Me}. For the solution of the isoperimetric problem in $\C P^{N-1}$ see also \cite[Theorem 4.1]{Me}. The isoperimetric sets are expected to transition from geodesic balls for small volumes to tubes around some $\C P^{M-1}\subset\C P^{N-1}$ for intermediate volumes.

\section{Faber--Krahn-type inequalities for the coherent state transform}\label{sec:fk}

The main result of the recent paper \cite{NiTi} by Nicola and Tilli states that, for any measurable set $E\subset\R^2$ of finite measure,
\begin{equation}
	\label{eq:niti}
	\iint_E |\langle \psi_{p,q},\psi\rangle|^2 \,dp\,dq \leq 2\pi\hbar \left( 1- e^{-(2\pi\hbar)^{-1}|E|} \right)
\end{equation}
with equality if and only if $\psi = e^{i\theta} \psi_{p_0,q_0}$ for some $p_0,q_0\in\R$, $\theta\in\R/2\pi\Z$ and $E$ is equal to a ball centered at $(p_0,q_0)$ (up to sets of measure zero). (We restrict ourselves here to the one-dimensional case of their result. Since the proof of Theorem \ref{heisen} extends to higher dimensions, the discussion in this subsection does so as well.)

We claim that the inequality \eqref{eq:niti} follows by abstract arguments from Theorem \ref{heisen}. Of course, this is not too surprising, since Kulikov's arguments, which we have adopted to yield a proof of Theorem \ref{heisen}, are inspired by those in \cite{NiTi}. Nevertheless, this observation will allow us to derive an analogue of the Nicola--Tilli results in the SU(2), SU(1,1) and $aX+b$ cases.

\begin{proof}[Proof of \eqref{eq:niti} given Theorem \ref{heisen}]
	Fixing $p_0,q_0\in\R$ and $\psi\in L^2(\R)$ with $\|\psi\|_{L^2(\R)}=1$, we can write the first assertion of Theorem \ref{heisen} as the statement that
	$$
	\iint_{\R\times\R} \Phi(|\langle\psi_{p,q},\psi\rangle|^2)\,dp\,dq \leq
	\iint_{\R\times\R} \Phi(|\langle\psi_{p,q},\psi_{p_0,q_0}\rangle|^2)\,dp\,dq
	$$
	for any convex function $\Phi$ on $[0,1]$. By Hardy--Littlewood majorization theory (see, e.g., \cite[Theorems 108, 249, 250]{HaLiPo}, \cite[Corollary 2.1]{AlTrLi}, \cite[Theorem 15.27]{Si3} and also \cite[Chapter 2, Propsition 3.3]{BeSh}), this is equivalent to the fact that
	$$
	\iint_E |\langle\psi_{p,q},\psi\rangle|^2 \,dp\,dq \leq \sup_{|F|=|E|}
	\iint_{F} |\langle\psi_{p,q},\psi_{p_0,q_0}\rangle|^2 \,dp\,dq
	$$
	for any measurable set $E\subset\R^2$ of finite measure. By an explicit computation, 
	$$
	|\langle\psi_{p,q},\psi_{p_0,q_0}\rangle| = e^{-\frac1{4\hbar}((q-q_0)^2+(p-p_0)^2)} \,.
	$$
	This is symmetric decreasing around $(p_0,q_0)$ and therefore the supremum above is attained if (and only if) $F$ is a ball centered at $(p_0,q_0)$ (up to sets of measure zero). In this case, the right side can be computed to be $1-e^{-(2\pi\hbar)^{-1}|E|}$, yielding \eqref{eq:niti}.
\end{proof}

By going carefully through the majorization argument it should be possible to deduce from the equality statement in Theorem \ref{heisen} the equality statment by Nicola and Tilli, but we omit this here.

Obviously, the above argument can be generalized to the SU(2), SU(1,1) and $aX+b$ cases. For the sake of brevity, we leave out a statement about the cases of equality.

\begin{theorem}\label{fksu2}
	Let $J\in\frac12\N$ and consider an irreducible $(2J+1)$-dimensional representation of $\mathrm{SU(2)}$ on $\mathcal H$. Then, for any $\psi \in \mathcal H$ with $\| \psi \|_{\mathcal H} = 1$ and any measurable $E\subset\Sph^2$,
	$$
	\int_{E} |\langle \psi_\omega, \psi \rangle|^2 \,d\omega \leq \frac{4\pi}{2J+1} \left( 
	1 - \left( 1-\tfrac{|E|}{4\pi} \right)^{2J+1} \right).
	$$
	Equality is attained if $\psi=e^{i\theta} \psi_{\omega_0}$ with some $\omega_0\in\Sph^2$, $\theta\in\R/2\pi\Z$ and $E$ is a spherical cap centered at $\omega_0$.
\end{theorem}

\begin{theorem}\label{fksu11}
	Let $K\in\frac12\N\setminus\{\tfrac12\}$ and consider the irreducible discrete series representation of $\mathrm{SU(1,1)}$ on $\mathcal H$ corresponding to $K$. Then, for any $\psi\in\mathcal H$ with $\|\psi\|_\mathcal H= 1$ and any measurable $E\subset\D$,
	$$
	\int_{E} |\langle \psi_z, \psi \rangle|^2 \,\frac{dA(z)}{(1-|z|^2)^2} \leq \frac{\pi}{2K-1} \left( 1- (1+m(E))^{-2K+1} \right),
	$$
	where $dm(z) = \pi^{-1} (1-|z|^2)^{-2}\,dA(z)$. Equality is attained if $\psi=e^{i\theta}\psi_{z_0}$ with some $z_0\in\D$, $\theta\in\R/2\pi\Z$ and $E$ is a hyperbolic ball centered at $z_0$.
\end{theorem}

\begin{theorem}\label{fkaxb}
	Let $\beta>\frac12$. Then, for any $\psi\in L^2(\R_+)$ with $\|\psi\|_{L^2(\R_+)} = 1$ and any measurable $E\subset\R_+\times\R$,
	$$
	\iint_{E} |\langle \psi_{a,b}, \psi \rangle|^2 \,\frac{da\,db}{a^2} \leq \frac{4\pi}{2\beta-1} \left( 1- (1+ (4\pi)^{-1} \mu(E))^{-2\beta+1} \right),
	$$
	where $d\mu(a,b) = a^{-2}da\,db$. Equality is attained if $\psi=e^{i\theta}\psi_{a_0,b_0}$ with some $a_0\in\R_+$, $b_0\in\R$, $\theta\in\R/2\pi\Z$ and $E$ is a hyperbolic ball centered at $(a_0,b_0)$.
\end{theorem}

In Theorems \ref{fksu11} and \ref{fkaxb} by a `hyperbolic ball' we mean a geodesic ball with respect to the hyperbolic metric on $\D$ and $\C_+$ (identified with $\R_+\times\R$), respectively. 

Theorem \ref{fkaxb}, including a characterization of equality cases, has recently been proved in \cite{RaTi} by a direct adaptation of the method in \cite{NiTi}. Our proof, based on Theorem \ref{axb}, is different.

\begin{proof}[Proof of Theorems \ref{fksu2}, \ref{fksu11} and \ref{fkaxb}]
	As above, one can show that the left sides in the theorems are bounded by the supremum of the integral of $|\langle \psi_\alpha,\psi_{\alpha_0}\rangle|^2$ over sets $F$ of the same measure as $E$. Here $\alpha$ means $\omega\in\Sph^2$, $z\in\D$ and $(a,b)\in\R_+\times\R$ in the three cases, respectively, and $\alpha_0$ is a fixed such index. By the bathtub principle, the supremum over $F$ is attained at a set of the form $\{ |\langle \psi_\alpha,\psi_{\alpha_0}\rangle|>\kappa_0\}\cup G$, where $G$ is a measurable subset of $\{ |\langle \psi_\alpha,\psi_{\alpha_0}\rangle|=\kappa_0 \}$.
	
	To complete the proof we will need some explicit knowledge about the function $|\langle \psi_\alpha,\psi_{\alpha_0}\rangle|$. We choose the representation space $\mathcal H$ in Theorems \ref{fksu2} and \ref{fksu11} in the same way as in the proofs of Theorems \ref{su2} and \ref{su11}, namely as $\mathcal P_{2J}$ and $A^2_{2K}(\D)$, respectively. Then, as shown there, $|\langle \psi_\alpha,\psi_{\alpha_0}\rangle| = u_{F_w}(z)$ where $\alpha=\mathcal S(z)$ and $\alpha=z$ in the first two cases and, similarly, $\alpha_0=\mathcal S(w)$ and $\alpha_0=w$. In the third case, if $\alpha=(a,b)$, then $z=\Sigma(ia-b)$ and similarly for $\alpha_0$ and $w$. In particular, $w=0$ if we choose $\alpha_0$ to be $\omega_0=(0,0,1)$, $z=0$ and $(a_0,b_0)=(1,0)$ in the different cases. The explicit definition of $u_f$ then shows that $\{ |\langle \psi_\alpha,\psi_{\alpha_0}\rangle|>\kappa_0\}$ is a spherical cap in the first case or a hyperbolic ball in the last two cases and that in all cases $\{ |\langle \psi_\alpha,\psi_{\alpha_0}\rangle|=\kappa_0\}$ has measure zero. Thus the set $G$ above can be ignored and we have identified the optimal set in the case of a special choice of $\alpha_0$.
	
	This, in fact, yields the shape for an arbitrary choice of $\alpha_0$. Indeed, as discussed before Theorem \ref{analytic}, the functions $\alpha\mapsto \langle \psi_\alpha,\psi_{\alpha_0}\rangle$ are equimeasurable for different $\alpha_0$'s and one such function can be obtained from another by the action of SU(2) or SU(1,1). Since this maps spherical caps to spherical caps, or hyperbolic balls to hyperbolic balls, we obtain that the supremum is attained for any $\alpha_0$ at such a set.
	
	It remains to compute the supremum. It is convenient to do this in terms of the functions $u_f$. The second case and the third case can be treated together with the convention that $\alpha=2K$ in the second and $\alpha=2\beta$ in the third case. We have with an arbitrary $F\in\mathcal M$
	$$
	\int_{\{ u_F>\kappa_0\}} u_F(z)^2\,dm(z) =
	\begin{cases}
		(4\pi)^{-1} \int_{\{ |\langle \psi_\omega,\psi_{\omega_0}\rangle| >\kappa_0\}} |\langle \psi_\omega,\psi_{\omega_0}\rangle|^2 d\omega \,, \\
		\pi^{-1} \int_{\{ |\langle \psi_z,\psi_{z_0}\rangle| >\kappa_0\}} |\langle \psi_z,\psi_{z_0}\rangle|^2 \frac{dA(z)}{(1-|z|^2)^2} \,.
	\end{cases}
	$$
	Meanwhile, by the layer cake formula,
	\begin{align}
		\label{eq:layercakefk}
		\int_{\{ u_F>\kappa_0\}} u_F(z)^2\,dm(z) & = 2 \int_0^{\kappa_0} m(\{ u_F>\kappa_0\})\kappa\,d\kappa + 2 \int_0^{\kappa_0} m(\{ u_F>\kappa\})\kappa\,d\kappa \notag \\
		& = \int_\Omega u_F(z)^2\,dm(z) - 2 \int_0^{\kappa_0}\!\! \left( m(\{ u_F>\kappa\}) - m(\{ u_F>\kappa_0\}) \right)\kappa\,d\kappa \,.
	\end{align}
	The first term on the right side is equal to
	$$
	\int_\Omega u_F(z)^2\,dm(z) = c\, \vvvert F \vvvert^2 = c
	$$
	with $c=(2J+1)^{-1}$ and $c=(\alpha-1)^{-1}$ in the different cases. For the second term on the right side of \eqref{eq:layercakefk} we use the explicit expressions for $m(\{ u_F>\kappa\})$ from the proof of Lemma \ref{mono} and get, after a computation,
	\begin{equation}
		\label{eq:kappa1}
		2 \int_0^{\kappa_0}\!\! \left( m(\{ u_F>\kappa\}) - m(\{ u_F>\kappa_0\}) \right)\kappa\,d\kappa =
		\begin{cases}
			c\,\kappa_0^\frac{2J+1}{J} \,,\\
			c\,\kappa_0^\frac{2(\alpha-1)}{\alpha} \,.
		\end{cases}
	\end{equation}
	This gives the expression of the supremum in terms of $\kappa_0$. The parameter $\kappa_0$ satisfies
	\begin{equation}
		\label{eq:kappa0}
		m(\{ u_F>\kappa_0\}) = 
		\begin{cases} 
			(4\pi)^{-1} |E| & \qquad\text{in the case of Theorem \ref{fksu2}} \,, \\
			m(E) & \qquad\text{in the case of Theorem \ref{fksu11}} \,, \\
			(4\pi)^{-1}\mu(E) & \qquad\text{in the case of Theorem \ref{fkaxb}} \,.
		\end{cases}
	\end{equation}
	(In the last case, we used the fact that $(4\pi)^{-1}\mu(E)= m(\Sigma(\tilde E))$, where $\Sigma$ is the conformal map from $\C_+$ to $\D$ from the proof of Theorem \ref{axb} and $\tilde E\subset\C_+$ is obtained from $E$ by identifying $(a,b)\in\R_+\times\R$ with $ia-b\in\C_+$.) Using \eqref{eq:kappa0} and the expressions from the proof of Lemma \ref{mono}, we can express $\kappa_0$ in terms of the measure of $E$. Inserting this into \eqref{eq:kappa1} gives an expression for the second term on the right side of \eqref{eq:layercakefk}. This leads to the claimed explicit form of the upper bound.
\end{proof}

\bibliographystyle{amsalpha}

\end{document}